\pdfoutput=1
\documentclass{article}
\usepackage{longtable,geometry}
\usepackage[latin1]{inputenc}
\usepackage{dsfont}
\usepackage{latexsym}
\usepackage{amsmath}
\usepackage{amsthm}
\usepackage{amssymb}

\usepackage{eurosym}

\usepackage{enumerate}

\usepackage{multirow}
\usepackage{slashbox}

\usepackage{epsf}
\usepackage{epsfig}

\usepackage{color}

\theoremstyle{plain}

\newtheorem{prop}{Proposition}[section]

\theoremstyle{definition}
\newtheorem{rem}{Remark}[section]
\newtheorem{lem}{Lemma}[section]
\newtheorem{defi}{Definition}[section]

\newsavebox{\fmbox}

\newcommand{\UneFigure}[8]{
\refstepcounter{figure}
\addcontentsline{lof}{figure}{\numberline{\thefigure}{\ignorespaces #5}}
\vspace*{#7cm}
\begin{center}
\begin{minipage}{#1cm}
\centerline{\includegraphics[width=#2cm,angle=#3]{#4}}
\vspace*{#8cm}
\begin{center}
\upshape{Figure \normalsize{\thefigure}:} #5
\end{center}
\label{#6}
\end{minipage}
\end{center}
}

\newcommand{\UnTableau}[5]{
\refstepcounter{table}
\addcontentsline{lof}{table}{\numberline{\thetable}{\ignorespaces #2}}
\vspace*{#4cm}
\begin{center}
{#1}
\vspace*{#5cm}
\begin{center}
\upshape{Table \normalsize{\thetable}:} #2
\end{center}
\label{#3}
\end{center}
}

\newcommand{\DeuxFiguresACote}[9]{
\refstepcounter{figure}
\addcontentsline{lof}{figure}{\numberline{\thefigure}{\ignorespaces #6}}
\begin{center}
\vspace*{#8cm}
\hspace*{-1cm} \begin{minipage}{#1\linewidth}
      \centering \includegraphics[width=#2cm,angle=#3]{#4}
   \end{minipage}\hfill
   \begin{minipage}{#1\linewidth}
      \centering \includegraphics[width=#2cm,angle=#3]{#5}
   \end{minipage} \hfill
\vspace*{#9cm}
\begin{center}
\upshape{Figure \normalsize{\thefigure}:} #6
\end{center}
\label{#7}
\end{center}
}

\newcommand{\Ccurs}{\begin{cal}C\end{cal}}

\newcommand{\Ecurs}{\begin{cal}E\end{cal}}
\newcommand{\Fcurs}{\begin{cal}F\end{cal}}

\newcommand{\Ocurs}{\begin{cal}O\end{cal}}

\newcommand{\Tcurs}{\begin{cal}T\end{cal}}

\newcommand{\E}{\mathbb{E}}
\newcommand{\F}{\mathbb{F}}

\newcommand{\Ind}{\mathds{1}}
\newcommand{\Proba}{\mathbb{P}}

\newcommand{\argmin}{\mathop{\mathrm{arg\,min}}}

\newcommand{\opt}{\mathop{\mathrm{opt}}}

\newcommand{\LL}{{\mathrm{L}}}

\newcommand{\pricing}{\mathop{\mathrm{pricing}}}

\title{A finite dimensional approximation for pricing moving average
  options}
\author{Marie Bernhart\footnote{Laboratoire de Probabilités et Modèles Aléatoires, CNRS, UMR 7599, Universités Paris 6-Paris 7
		and EDF R$\&$D, 92141 Clamart, France. Email: \texttt{marie-externe.bernhart@edf.fr}}
	\and Peter Tankov\footnote{Centre de
    Math\'ematiques Appliqu\'ees, Ecole Polytechnique, 91128
    Palaiseau, France. Email: \texttt{peter.tankov@polytechnique.org}}
  \and Xavier Warin\footnote{EDF R$\&$D, 92141 Clamart, France and Laboratoire de Finance des Marchés de l'Energie, Université Paris Dauphine. Email: \texttt{xavier.warin@edf.fr}}
  }

\date{}

\begin{document}

\maketitle

\begin{abstract}
We propose a method for pricing American
options whose pay-off depends on the moving average of the underlying
asset price. The method uses a finite dimensional approximation of the
infinite-dimensional dynamics of the moving average process based on a
truncated Laguerre series expansion.
 The resulting problem
is a finite-dimensional optimal stopping problem, which we propose to
solve with a least squares Monte Carlo approach. We analyze the
theoretical convergence rate of our method and present numerical
results in the Black-Scholes framework.

\end{abstract}



%%%%%%%%%%%%%%%%%%%%%%%%%%%%%%%%%%%%%%%%%%%%%%%%%%%%%%%%
%
%%%%%%%%%%%%%%%%%%%%%%%%%%%%%%%%%%%%%%%%%%%%%%%%%%%%%%%%
\section{Introduction}

We introduce a new method to value American options whose payoff at
exercise depends on the moving average
of the underlying asset price. The simplest example (sometimes known as
surge option) is a variable strike call or put, whose strike is adjusted daily to the
moving average of the underlying asset over a certain fixed-length
period preceding the current date. American-style options on moving average are widely used in energy markets.
In gas markets, for example, these options are known as indexed Swing
options and allow the holder
to purchase an amount of gas at a strike price, which is indexed on
moving averages of various oil-prices:
%MB
typically gas oil and fuel oil prices are averaged over the last 6 months and delayed in time with a 1 month lag.

We shall denote by $X$ the moving average of the underlying $S$ over a time window with fixed length $\delta > 0$:
\begin{eqnarray}
X_t = \frac{1}{\delta} \int_{t - \delta}^{t} S_{u} du, \quad \forall t \geq \delta.
\label{def-X}
\end{eqnarray}
The process $X$ follows the dynamics
\[ \begin{array}{c}
	%X_{\delta} = \frac{1}{\delta} \int_{0}^{\delta} S_{u} du
	d X_t = \frac{1}{\delta} \left( S_t - S_{t - \delta} \right)
        dt, \quad \forall t \geq \delta.
\end{array} \]
This shows in particular that even if $S$ is Markovian, the process $(S, X)$ is not:
it is, in general, impossible for any finite $n$ to
find $n$ processes $X^1,\ldots, X^n$ such that $(S, X,X^1,\ldots,X^n)$ are
jointly Markovian. This property makes the pricing of the moving window options
with early exercise a challenging problem both from the theoretical and the
numerical viewpoint. In a continuous time framework the problem is infinite dimensional,
 and in a discrete time framework (pricing of a Bermudan option instead of
an American option) there is a computational challenge, due to high
dimensionality: the dimension is equal to the number of time steps
within the averaging window. This in particular makes it difficult to compute the conditional expectations involved in the optimal exercise rule.

The problem of pricing moving average American options should not be
confused with a much simpler problem of pricing Asian American
options with a fixed start averaging window, where the payoff depends on
$$
A_t = \frac{1}{t} \int_{0}^{t} S_{u} du, \forall t > 0.
$$
It is well-known (see for example Wilmott and al. \cite{WHD96}) that in this case, adding a dimension to the problem allows to derive a finite dimensional Markovian formulation.
On the other hand, partial average Asian options of European style can
be easily valued (see for example Shreeve \cite{She97}). If the
averaging period has a length $\delta > 0$, then on $[T - \delta, T]$
the option value is given by the price of the corresponding Asian
option and on $[0, T - \delta]$ it solves a European style PDE with appropriate terminal and boundary conditions.

In this paper, we propose a method for pricing moving average American
options based on a finite dimensional approximation of the
infinite-dimensional dynamics of the moving average process. The
approximation is based on a truncated expansion of the weighting
measure used for averaging in a series involving
Laguerre polynomials.
This technique has long been used in signal processing (see for example \cite{Mal90}) but is less known in the context of approximation
of stochastic systems.
We compute the rate of convergence of our method as
function of the number of terms in the series.  The resulting problem is then a finite-dimensional optimal stopping problem, which we propose to solve with a Monte Carlo Longstaff and Schwartz-type approach.
%MB
Numerical results are presented for moving average options in the Black-Scholes framework.
%commonly encountered in electricity and gas markets.

%OVERVIEW LITERATURE

In the literature, very few articles discuss moving average
options with early exercise feature \cite{Bil03,BC07,DLZ09,Gra08,KL03}. A common approach (see e.g., Broadie
and Cao \cite{BC07}) is to use the least squares Monte Carlo,
computing the conditional expectation estimators through regressions on
polynomial functions of the current values of the underlying price and its moving average.  Since the future evolution of the moving average
depends on the entire history of the price process between $t-\delta$
and $t$, this approach introduces a bias. In our numerical examples we
compare this approach to our results, and find that for standard
moving average American options the error is not so large (less than
1\% for the examples we took), which justifies the use of this
approach for practical purposes in spite of its suboptimality. For
moving average American options with time delay, whose payoff depends
on the average of the price between dates $t -\delta_1$ and
$t-\delta_2$, $0<\delta_2<\delta_1$, the suboptimal approximation
leads to a bias of up to $11\%$ of the option's price in our examples.

Bilger \cite{Bil03} uses a regression based approach in the
discrete-time setting to compute the conditional expectations
considering that the state vector is composed of the underlying price, its moving average and additional partial averages of the price over the rolling period. Their number is computed heuristically and as it tends to the number of time steps within the rolling period, the computed price tends to the true price of the moving average option.
The same kind of approach is used by Grau \cite{Gra08}, but the author
improves its numerical efficiency by a different choice of basis
functions in the regressions used for the conditional expectations
estimation.

Kao and Lyuu \cite{KL03} introduce a tree method based on the CRR
model to price moving average lookback and reset options. Their method can handle only short averaging windows: the numerical results that are shown deal at most with $5$ discrete observations in the averaging period. Indeed, this tree-based approach leads to an algorithm complexity (number of tree nodes) which exponentially increases with the number of time steps in the averaging period.
Finally, Dai et al.~\cite{DLZ09} introduce a lattice algorithm for pricing Bermudan moving average barrier options. The authors propose a finite dimensional PDE model for such options and solve it using a grid method.

The pricing of moving average options is closely related to high-dimensional optimal stopping problems.
It is well-known that deterministic techniques such as finite
differences or approximating trees are made inefficient by the
so-called curse of dimensionality. Only Monte Carlo type techniques
can handle American options in high dimensions. Bouchard and Warin
\cite{BW10} and references therein shall give to the interested reader
a recent review of this research field.

More generaly, a related problem is that of optimal stopping of
stochastic differential equations with delay. With the exception of a
few cases where explicit dimension reduction is possible
\cite{federico.oksendal.10,gapeev.reiss.06}, there is no numerical
method for solving such problems, and the Laguerre approximation
approach of the present paper is a promising direction for
further research.

The rest of the paper is structured as follows. In Section
\ref{sec-part-theo}, we introduce the mathematical context and provide a general result which links the strong error of
approximating one moving average process with another to a certain
distance between their weighting measures. We then introduce an approximation of the weighting measure as a series of Laguerre
functions truncated at $n$ terms, which leads to $(n+1)$-dimensional Markovian approximation
to the initial infinite dimensional problem.
The properties of Laguerre functions combined with our strong
approximation result then enable us to establish a bound on the pricing error introduced by our approach as $n$ goes to infinity.
In Section \ref{sec-method-num}, our numerical method, based on least
squares Monte Carlo algorithm, is presented.
The final section of this paper reports the results of numerical
experiments in the Black-Scholes framework which include pricing 
moving average options with time delay.
%MB
%and of multi-asset options in the gas market.

Throughout the paper we assume that the price of the underlying asset
$S = (S_t)_{t \geq 0}$ is a non-negative continuous Markov process
defined on the probability space $(\Omega, \Fcurs,\F,\Proba)$,
where $\Proba$ is a martingale probability for the
financial market and $\F = (\Fcurs_t)_{t \leq T}$ is the natural
filtration of $S$.  

For the sake of simplicity, we present our results in the
framework of a $1$-dimensional price model but they are directly
generalizable to a multi-asset model or to a model with unobservable
risk factors such as stochastic volatility.

As usual, we denote by $\LL^{2} \equiv \LL^{2}\left( [0, + \infty) \right)$ the Lebesgue space of real-valued square-integrable functions $f$ on $[0, + \infty)$ endowed with its norm:
$$
\left\| f \right\|_{2} := \left[ \int_{0}^{\infty} \left| f(x) \right|^{2} dx \right]^{\frac{1}{2}}.
$$

%%%%%%%%%%%%%%%%%%%%%%%%%%%%%%%%%%%%%%%%%%%%%%%%%%%%%%%%
%
%%%%%%%%%%%%%%%%%%%%%%%%%%%%%%%%%%%%%%%%%%%%%%%%%%%%%%%%
\section{\large A finite dimensional approximation of moving average options price}
\label{sec-part-theo}

\paragraph{Strong approximations of moving average processes} %optimal-scaled
\label{subsec-error-general-MM-pb}

Consider a general moving average process of the form\footnote{In the
  literature (see \cite{basse.pedersen.09} and references therein),
  moving averages are usually defined via the stochastic integral of
  $S$. Our definition as an ordinary integral with respect to
a weighting measure is closer to the financial specifications. }
$$
M_t = \int_{0}^\infty  S_{t-u} \mu(du)
$$
where $\mu$ is a finite possibly signed measure on $[0,\infty)$.
Throughout the paper, we shall adopt the following convention for the values of $S$ on the negative time-axis:
\begin{equation}
	S_t = S_0,\quad  \forall t \leq 0.
\label{convention-S}
\end{equation}

We shall use an integrability assumption on the modulus of continuity of
 the price process: there exists a constant $C<\infty$ such that
\begin{align}
\E \left[\sup_{t,s\in [0,T]: |t-s|\leq h} |S_t -S_s|\right] \leq C
\varepsilon(h),\quad \varepsilon(h):=
\sqrt{h \ln \left(\frac{2T}{h}\right)}. \label{modcont.eq}
\end{align}
Fischer and Nappo \cite{fischer.nappo.10} show that this holds in particular when $S$ is a continuous Itô
process of the form
$$
S_t = S_0 + \int_0^t b_s ds + \int_0^t \sigma_s dW_s
$$
with
$$
\E \left[\sup_{0\leq t \leq T} |b_s|\right]<\infty \quad \text{and}\quad \E \left[\sup_{0\leq t \leq T} |\sigma_s|^{1+\gamma}\right]<\infty
$$
for some $\gamma>0$.

The following lemma provides a tool for comparing moving averages with
different weighting measures.
\begin{lem}\label{comp.lm}
Let Assumption \eqref{modcont.eq} be satisfied, and let $\mu$ and $\nu$ be
finite signed measures on $[0,\infty)$ with Jordan decompositions $\mu
= \mu^+ - \mu^-$ and $\nu
= \nu^+ - \nu^-$, such that $\mu^+(\mathbb R_+)>0$.
Define
$$
M_t = \int_{0}^\infty  S_{t-u} \mu(du),\qquad N_t = \int_{0}^\infty  S_{t-u} \nu(du).
$$
Then
\begin{multline}
\E \left[\sup_{0\leq t \leq T} |M_t - N_t|\right] \leq
C|\mu(\mathbb R_+)-\nu(\mathbb R_+)| \\+
C \Big( \mu^+([0,T])+\nu^-([0,T])+|\mu([0,T])-\nu([0,T])| \Big) \
\varepsilon\left(\frac{1}{\mu^+([0,T])} \int_0^T |F_\mu(t) - F_\nu(t)|dt\right)\label{distbound}
\end{multline}
for some constant $C<\infty$ which does not depend on $\mu$ and $\nu$,
where
$$
F_\nu(t) := \nu([0,t])\quad \text{and}\quad F_\mu(t) := \mu([0,t]).
$$
\end{lem}
\begin{proof}[\underline{\textbf{Proof}}]
\noindent\textit{Step 1.}\quad We first assume that $\mu$ and $\nu$ are
probability measures. Let $F^{-1}_{\mu}$ and $F^{-1}_{\nu}$ be
generalized inverses of $\mu$ and $\nu$ respectively. Then,
\begin{align*}
\E\left[\sup_{0\leq t \leq T} |M_t - N_t|\right] &= \E \left[\sup_{0\leq t
    \leq T} \int_0^1 |S_{t-F^{-1}_\mu(u)} -
  S_{t-F^{-1}_\nu(u)}|du\right]\\
& \leq \int_0^1 \E \left[\sup_{0\leq t
    \leq T} |S_{t-F^{-1}_\mu(u)} -
  S_{t-F^{-1}_\nu(u)}|\right] du\\
& \leq C \int_0^1 \varepsilon
\left(|F^{-1}_\mu(u)\wedge T-F^{-1}_\nu(u) \wedge T|\right) du \\
& \leq C \varepsilon \left(\int_0^1 |F^{-1}_\mu(u)\wedge T-F^{-1}_\nu(u) \wedge T|du \right),
\end{align*}
where the last inequality follows from the concavity of
$\varepsilon(h)$. The expression inside the brackets is the
Wasserstein distance between the measures $\mu$ and $\nu$ truncated at
$T$. Therefore, from the Kantorovich-Rubinstein theorem we deduce
$$
\E \left[\sup_{0\leq t \leq T} |M_t - N_t|\right] \leq C \varepsilon \left(\int_0^T |F_\mu(t)-F_\nu(t) |dt \right).
$$

\noindent\textit{Step 2.}\quad Introduce $\tilde \mu = \mu 1_{[0,T]}$ and $\tilde \nu = \nu 1_{[0,T]} +
(\mu([0,T])-\nu([0,T]))\delta_{2T}$, where $\delta_{2T}$
is the point mass at the point $2T$. Then,
\begin{align*}
& \E\left[\sup_{0\leq t \leq T} |M_t - N_t|\right] \leq C|\mu(\mathbb
R_+) - \nu(\mathbb R_+)| + \E \left[\sup_{0\leq t
    \leq T} \left|\int_0^\infty S_{t-u}\tilde\mu(du) -\int_0^\infty
    S_{t-u}\tilde\nu(du)\right|\right]\\
&\qquad \leq C|\mu(\mathbb
R_+) - \nu(\mathbb R_+)| \\ &+ (\tilde\mu^+(\mathbb R_+)+\tilde\nu^-(\mathbb
      R_+)) \E \left[\sup_{0\leq t
    \leq T} \left|\int_0^\infty
    S_{t-u}\frac{\tilde\mu^+(du) + \tilde\nu^-(du)}{\tilde\mu^+(\mathbb R_+)+\tilde\nu^-(\mathbb  R_+)} -\int_0^\infty
    S_{t-u}\frac{\tilde\mu^-(du) + \tilde\nu^+(du)}{\tilde\mu^+(\mathbb R_+)+\tilde\nu^-(\mathbb  R_+)}\right|\right].
\end{align*}
Since $\tilde\mu^+(\mathbb R_+)+\tilde\nu^-(\mathbb R_+) = \tilde\mu^-(\mathbb R_+)+\tilde
\nu^+(\mathbb R_+)$, both measures under the
integral sign are probability measures, and we can apply Step 1, which
gives
\begin{align*}
& \E\left[\sup_{0\leq t \leq T} |M_t - N_t|\right] \leq C|\mu(\mathbb
R_+) - \nu(\mathbb R_+)|\\
&\qquad +  C(\tilde\mu^+(\mathbb R_+)+\tilde\nu^-(\mathbb R_+))
\varepsilon\left(\frac{1}{ \tilde\mu^+(\mathbb R_+)+\tilde\nu^-(\mathbb
    R_+)}\int_0^T |F_{\tilde\mu^+ + \tilde \nu^-}(t) - F_{\tilde\mu^- + \tilde
    \nu^+}(t)|dt\right)\\
& = C|\mu(\mathbb
R_+) - \nu(\mathbb R_+)| +  C(\tilde\mu^+(\mathbb R_+)+\tilde\nu^-(\mathbb R_+))
\varepsilon\left(\frac{1}{ \tilde\mu^+(\mathbb R_+)+\tilde\nu^-(\mathbb
    R_+)}\int_0^T |F_{\mu}(t) - F_{\nu}(t)|dt\right),
\end{align*}
because $\tilde \mu$ coincides with $\mu$ and $\tilde \nu$ coincides with $\nu$ on $[0,T]$. Using the
properties of the function $\varepsilon$ and the definition of $\tilde
\mu$ and $\tilde \nu$, we then get
\eqref{distbound} with a different constant $C$.
\end{proof}

%Let us now limit ourselves to the more usual situation when the
%measure $\mu$ has a density denoted by $h$, which is assumed to be
%square integrable on $[0,T]$. In this case the moving
%average
%\begin{align}
%M_t = \int_{0}^\infty  S_{t-u} h(u)du\label{ma.def}
%\end{align}
%is a continuous finite variation process (see e.g.,
%\cite{reiss.al.07}).

\paragraph{Introducing Laguerre approximation}
The aim of this paragraph is to provide heuristic arguments which lead
to Laguerre approximation of the moving average. 
%A rigorous
%justification with convergence rate is given in Proposition \ref{prop-erreur-mm}.
We would like to find a finite-dimensional approximation to
$M$, that is, find $n$ processes $Y^1,\ldots,Y^n$ such that
$(S,Y^1,\ldots,Y^n)$ are jointly Markov, and $M_t$ is approximated in
some sense to be made precise later by $M^n_t$ which depends
deterministically on $S_t,Y^1_t, \ldots, Y^n_t$.

Since $M$ is linear in $S$, it
is natural to require that the approximation also be linear. Therefore, we assume that $Y =
(Y^1,\ldots,Y^n)$ satisfies the linear SDE
\begin{align}
dY_t = - A Y dt + \mathbf{1} (\alpha S_t dt + \beta d S_t) ,\label{sde}
\end{align}
where $A$ is an $n\times n$ matrix, $\mathbf 1$ is a
$n$-dimensional vector with all components equal to $1$ and $\alpha$
and $\beta$ are constants. Similarly,
the approximation is given by a linear combination of the components of
$Y$: $M^n = B^\perp Y$, where $B$ is a vector of size $n$ and $\perp$
denotes the matrix transposition.

The solution to \eqref{sde} can be written as
$$
Y_t = e^{-At} Y_0 + \int_0^t e^{-A(t-s)}\mathbf{1} (\alpha S_s ds + \beta d S_s)
$$
or, assuming stationarity, as
$$
Y_t = \int_{-\infty}^t e^{-A(t-s)}\mathbf{1} (\alpha S_s ds + \beta d S_s) \quad
\text{and}\quad M^n_t = \int_{-\infty}^t B^\perp e^{-A(t-s)}\mathbf{1} (\alpha S_s ds + \beta d S_s).
\quad
$$
Integration by parts then yields:
$$
M^n_t = \beta B^\perp \mathbf{1} S_t +  \int_{-\infty}^t B^\perp
(\alpha - A \beta) e^{-A(t-s)}\mathbf{1}S_s ds:= K_n S_t + \int_{-\infty}^t h_n(t-u) S_u du,
$$
Recalling the structure of the matrix exponential, it follows that the function $h_n$ is of the form
\begin{align}
h_n(t) = \sum_{k=1}^{K} e^{-p_k t} \sum_{i=0}^{n_k} c_i^k t^i,\label{hankel}
\end{align}
where $n_1 + \ldots + n_K + K = n$ ($K$ is the number of Jordan blocks
of $A$).
Therefore, the problem of finding a finite-dimensional approximation
for $M$ boils down to finding an approximation of the form $K_n
\delta_{0}(dt) + h_n(t) dt$ for the measure $\nu$. This problem is well known in
signal processing, where the density $h$ of $\nu$ is called impulse
response function of a system, and $h_n$ is called Hankel
approximation of $h$. For arbitrary $\mu$ and $n$, Hankel approximations
may be very hard to find, and in this paper we shall focus on a
subclass for which $K=1$, that is, the function $h_n$ is of the form
\begin{align}
h_n(t) =  e^{-p t} \sum_{i=0}^{n-1} c_i t^i.\label{laguerre}
\end{align}
This is known as Laguerre approximation, because for a fixed
$p$, the first $n$ scaled Laguerre
functions (defined below) form an orthonormal basis of the space of
all functions of the form \eqref{laguerre} endowed with the scalar
product of $\LL^2([0,\infty))$ which will be denoted by $\left\langle \cdot , \cdot \right\rangle$. See \cite{maekilae.96} for a discussion
of optimality of Laguerre approximations among all approximations of
type \eqref{hankel}.
%MB
%Useful properties of Laguerre polynomials are provided in Appendix.

\begin{defi} Fix a scale parameter $p > 0$. The scaled Laguerre functions $(L^{p}_k)_{k \geq 0}$ are defined on $[0, + \infty)$ by
\begin{equation}
L^{p}_k (t) = \sqrt{2p} \ P_k (2pt) e^{- p t}, \quad \forall k \geq 0
\label{def-laguerre-parametre}
\end{equation}
in which $(P_k)_{k \geq 0}$ is the family of Laguerre polynomials explicitly defined on $[0, + \infty)$ by
\begin{equation}
P_k (t) = \sum_{i = 0}^{k} {k \choose k - i} \frac{(-t)^{i}}{i !},
\quad \forall k \geq 0
\label{def-explicite-laguerre}
\end{equation}
or recursively by
\begin{equation}
\begin{cases}
P_{0}(t) = 1 \\
P_{1}(t) = 1 - t \\
P_{k + 1}(t) = \frac{1}{k + 1} \left( (2k + 1 - t)P_k(t) - k P_{k - 1}(t)\right), \forall k \geq 1.
\end{cases}
\label{prop-rec-laguerre}
\end{equation}
The scaled Laguerre functions $(L^{p}_k)_{k \geq 0}$ form an
orthonormal basis of the Hilbert space $\LL^2([0,\infty))$, i.e.
$$
\forall (j, k),\quad  \left\langle L^{p}_j, L^{p}_k \right\rangle = \delta_{j, k}.
$$
\end{defi}
Fix now an order $n \geq 1$ of truncation of the series. In view of
Lemma \ref{comp.lm}, we propose the following Laguerre approximation
of the moving average process $M$:
\begin{itemize}
\item Let $H(x) = \mu([x,+\infty))$.
\item Compute the Laguerre coefficients of the function $H$:
$$
A^p_k = \langle H, L^p_k \rangle.
$$
Set $H^p_n(t) = \sum_{k=0}^{n-1} A^p_k L^p_k(t)$ and $h^p_n(t) =
-\frac{d}{dt}H^p_n(t)$. In view of Lemma \ref{deriv.lm}, the function
$h^p_n$ can be written as
\begin{eqnarray}
h^p_n  = \sum_{k=0}^{n-1} a^{p}_k L^{p}_k,\qquad a^p_k = pA^p_k + 2p
\sum_{i=k+1}^{n-1} A^p_i.
\label{def-approx-I}
\end{eqnarray}
\item Approximate the moving average $M$ with
\begin{eqnarray}
M^{n, p}_t = (H(0)-H^p_n(0))S_t + \int_{0}^{+\infty} h^{p}_n (u) S_{t-u} du, \quad \forall t \geq 0.
\label{def-approx-X}
\end{eqnarray}
\end{itemize}

\begin{rem} The approximation proposed in \eqref{def-approx-X} (and in particular the correction coefficient in front of $S_t$) is chosen so that the total mass of the weighting measure of the approximate moving average $M^{n, p}_t$ is equal to the total mass of the weighting measure $\mu$ of the exact moving average. In particular, such an approximation becomes exact for a constant asset price $S$.
\end{rem}

From definitions \eqref{def-approx-I} and \eqref{def-approx-X}, it seems natural to introduce $n$ random processes $X^{p, 0}, \ldots, X^{p, n-1}$ defined by
\begin{eqnarray}
X^{p, k}_t =  \int_{0}^{+\infty} L^{p}_k(v) S_{t-v} dv, \forall t \geq 0, \forall k = 0, \ldots, n-1.
\label{def-Xk}
\end{eqnarray}
They will be called \textit{Laguerre
  processes} associated to the process $S$ throughout this paper and are related to the moving average approximation by
\begin{eqnarray}
M^{n, p}_t = (H(0)-H^p_n(0))S_t +\sum_{k=0}^{n-1} a^{p}_k X^{p, k}_t,\quad \forall t \geq 0.
\label{def-approx-X-sum}
\end{eqnarray}

%%%%%%%%%%%%%%%%%%%%%%%%%%%%%%%%%%%%%%%%%
%
%%%%%%%%%%%%%%%%%%%%%%%%%%%%%%%%%%%%%%%%%
\begin{prop} Let $n \geq 1$ and $p > 0$. The $(n+1)$-dimensional process
  $(S, X^{p, 0}, X^{p, 1},
  \ldots, X^{p, n-1})$ is Markovian. The $n$ Laguerre processes
  follow the dynamics:
\begin{equation*}
	\begin{cases}
	d X^{p, 0}_t = \left( \sqrt{2p} S_t - p X^{p, 0}_t \right) dt \\
	d X^{p, 1}_t = \left( \sqrt{2p} S_t - 2 p X^{p, 0}_t - p X^{p, 1}_t \right) dt \\
	%d X^{p, 2}_t = \left( \sqrt{2p} S_t - p \left( 2 X^{p, 0}_t + 2 X^{p, 1}_t + X^{p, 2}_t \right) \right) dt \\
	\vdots \\
	d X^{p, n-1}_t = \left( \sqrt{2p} S_t - 2 p \sum^{n-2}_{k = 0} X^{p, k}_t - p X^{p, n-1}_t \right) dt \\
	\end{cases}
\end{equation*}
with initial values
\begin{equation}
X^{p, k}_0 = S_0 (-1)^{k} \frac{\sqrt{2p}}{p}, \forall k \geq 0.
\label{def-val-ini-Xk}
\end{equation}
\label{prop-approx-pb-markov}
\end{prop}

%%%%%%%%%%%%%%%%%%%%%%%%%%%%%%%%%%%%%%%%%
\begin{proof}[\underline{\textbf{Proof}}] Immediate, from Equations
  \eqref{def-Xk} and \eqref{def-laguerre-parametre} and properties \ref{lem-lag}-\eqref{prop-Lag-a} and \ref{lem-lag}-\eqref{prop-Lag-b}.
\end{proof}

%\begin{rem} In the framework of a multi-assets model as introduced in Remark \ref{rem-dim-S}, the same kind of approximation based on a finite Laguerre expansion can be formulated. The resulting dimension of the approximate problem is $(d+n)$. In addition, the same estimate of the pricing error as in proposition \ref{prop-erreur-pricing} is obtained as soon as $\phi$ is Lipschitz in its $(d+1)$-th variable.
%\label{rem-multi-asset}
%\end{rem}

\paragraph{Convergence of the Laguerre approximation}

\begin{prop} Let Assumption \eqref{modcont.eq} be
  satisfied, and suppose that the moving average process $M$ is of the
  form
\begin{align}
M_t = K_0 S_t + \int_{0}^\infty  S_{t-u} h(u)du
\label{ma.def}
\end{align}
where $K_0$ is a constant and the function $h$ has compact support, finite variation on
  $\mathbb R$, is constant in the neighborhood of zero and
  is not a.e. negative on $[0,T]$. Then the error of approximation
  \eqref{def-approx-X-sum}
admits the bound
$$
\mathbb E \left[ \sup_{0\leq t\leq T} |M_t - M^{n,p}_t| \right] \leq C \varepsilon(n^{-\frac{3}{4}}).
$$
\label{prop-erreur-mm}
\end{prop}
%
%%%%%%%%%%%%%%%%%%%%%%%%%%%%%%%%%%%%%%%%%
\begin{proof}[\underline{\textbf{Proof}}]
We shall use Lemma \ref{comp.lm}. The
measures $\mu$ and $\nu$ are defined by $\mu(dx) = K_0 \delta_0(dx) + h(x) dx$ and
$\nu(dx) = (H(0)-H^p_n(0))\delta_0(dx) + h^p_n(x) dx$. Therefore, these
measures have the same mass, and the first term in estimate
\eqref{distbound} disapears. In addition,
$$
|\mu([0,T])-\nu([0,T])| = |H(T) - H^p_n(T)|,
$$
which remains bounded by Lemma \ref{coefs.lm}. Let us show that
$\nu^-([0,T])$ is bounded as well. For this it is enough to prove that
$\|(H^p_n)'\|_2$ is bounded on $n$. A straightforward computation using
Lemma \ref{deriv.lm} shows that
\begin{eqnarray}
c^p_k & = & \sqrt{2p} H(0) - 2p \sum_{i = 0}^{k-1} A^p_i - p A^p_k,
\label{lien-c-A}
\end{eqnarray}
where $c^p_k:=\langle h, L^p_k\rangle$ are Laguerre coefficients of $h$. By definition of $H^p_n$ and $a^p_k$ in \eqref{def-approx-I}, this leads to
\[ \begin{array}{rcl}
a^p_k & = & c^p_k - \sqrt{2p} \left[ H(0) - H^{p}_n (0) \right].
\end{array} \]
We have thus:
$$
\|(H^p_n)'\|^2_2 = \sum_{k\leq n-1} |a^p_k|^2
%= \sum_{k\leq n-1} \left| c^p_k - \sqrt{2p} \left( H(0) - H^{p}_n (0) \right) \right|^2
\leq 2\sum_{k\leq n-1} \left|c^p_k\right|^2 + 2 \sum_{k\leq n-1} \left| \sqrt{2p} \left[ H(0)-H^p_{n}(0) \right] \right|^2 = \Ocurs (n^{-\frac{1}{2}})
$$
by Lemma \ref{coefs.lm} and using $\sqrt{2p} \left[ H(0)-H^p_{n}(0) \right] = c^p_n + p A^p_n$ issued from \eqref{lien-c-A}.
%Both terms are then bounded uniformly on $n$ by Lemma \ref{coefs.lm}.
%Lemma \ref{coefs.lm} implies that the second sum is in $\Ocurs \left( n^{-\frac{1}{2}} \right)$
Therefore, there
exists a constant $C<\infty$, which does not depend on $n$, such that
$$
\mathbb E[\sup_{0\leq t\leq T} |M_t - M^{n,p}_t|] \leq C
\varepsilon\left(\frac{1}{\int_0^T h^+(t) dt} \int_0^T |H(t) - H^p_n(t)|dt\right).
$$
By Cauchy-Schwartz inequality and Lemma \ref{coefs.lm},
$$
\int_0^T |H(t) - H^p_n(t)|dt \leq \sqrt{T}\|H-H^p_n\|_2 = \sqrt{T} \left( \sum_{k\geq
  n} \left|A^p_k\right|^2 \right)^{\frac{1}{2}} = \Ocurs(n^{-\frac{3}{4}}),
$$
from which the result follows using the properties of $\varepsilon$
and the fact that $h$ is not a.e. negative on $[0,T]$ (which means
that $\int_0^T h^+(t)dt>0$).

\end{proof}

\paragraph{Approximating option prices}
The price of the American option whose pay-off depends on the moving
average $M$ and the price of the underlying is given by
%\sup_{\tau \in \Tcurs_{[\delta, T]}} \E \left[ \phi \left( S_{\tau}, X_{\tau} \right) \right]
%\label{def-pb}
%\end{eqnarray}
%where $\Tcurs_{[\delta, T]}$ is the set of $\F$-stopping times
%valued in $[\delta, T]$, $\phi$ is the payoff function
$$
\sup_{\tau \in \Tcurs} \E \left[ \phi \left( S_{\tau}, M_{\tau} \right) \right]
$$
where $\Tcurs$ is the set of $\F$-stopping times
and $\phi$ is the payoff function. It can then be approximated
by the solution to
\begin{eqnarray}
\sup_{\tau \in \Tcurs} \E \left[ \phi \left( S_{\tau}, M^{n, p}_{\tau} \right) \right].
\label{def-approx-pb}
\end{eqnarray}
%in which we recall that:
%\begin{itemize}
%	\item $n \geq 1$ is the order of the approximation (number of Laguerre functions introduced)
	%, which is aimed at tending to infinity.
%	\item $p > 0$ is the scale parameter of the Laguerre functions, which can be chosen to minimize the $\LL^2$-distance between $H$ and $H^{p}_n$% (see section \ref{sec-method-num}).
%\end{itemize}

\begin{prop} Let Assumption \eqref{modcont.eq} be
  satisfied, and suppose that the payoff function $\phi$ is Lipschitz
  in the second variable and that the moving average process $M$
  satisfies the assumptions of Proposition \ref{prop-erreur-mm}. Then the pricing error admits the bound
\begin{eqnarray*}
\Ecurs_{\pricing} (n, p)  := \left|	\ \sup_{\tau \in \Tcurs} \E \left[ \phi \left( S_{\tau}, M_{\tau} \right)
  \right] - \sup_{\tau \in \Tcurs} \E \left[ \phi \left(
      S_{\tau},  M^{n, p}_{\tau} \right) \right] \right|\leq C \varepsilon(n^{-\frac{3}{4}}).
\end{eqnarray*}
where $C > 0$ is a constant independent of $n$.
\label{prop-erreur-pricing}
\end{prop}
%
%%%%%%%%%%%%%%%%%%%%%%%%%%%%%%%%%%%%%%%%%
\begin{proof}[\underline{\textbf{Proof}}]
We have first:
\begin{eqnarray*}
	\forall \tau, \E \left[ \phi \left( S_{\tau}, M_{\tau} \right) \right] & = & \E \left[ \phi \left( S_{\tau},   M^{n, p}_{\tau} \right) \right] + \E \left[ \phi \left( S_{\tau}, M_{\tau} \right) - \phi \left( S_{\tau},   M^{n, p}_{\tau} \right) \right] \\
	\Longrightarrow \
	\sup_{\tau} \E \left[ \phi \left( S_{\tau}, M_{\tau} \right) \right] & = & \sup_{\tau} \Big( \E \left[ \phi \left( S_{\tau},    M^{n, p}_{\tau} \right) \right] + \E \left[ \phi \left( S_{\tau}, M_{\tau} \right) - \phi \left( S_{\tau},    M^{n, p}_{\tau} \right) \right] \Big) \\
	& \leq & \sup_{\tau} \E \left[ \phi \left( S_{\tau},    M^{n, p}_{\tau} \right) \right] + \sup_{\tau} \E \left[ \phi \left( S_{\tau}, M_{\tau} \right) - \phi \left( S_{\tau},    M^{n, p}_{\tau} \right) \right] 		
\end{eqnarray*}
In consequence,
\begin{eqnarray*}
	\sup_{\tau} \E \left[ \phi \left( S_{\tau}, M_{\tau} \right) \right] - \sup_{\tau} \E \left[ \phi \left( S_{\tau},    M^{n, p}_{\tau} \right) \right]
		%& \leq & \sup_{\tau} \E \left[ \phi \left( S_{\tau}, M_{\tau} \right) - \phi \left( S_{\tau},    M^{n, p}_{\tau} \right) \right] \\
		& \leq & \sup_{\tau} \E \left| \phi \left( S_{\tau}, M_{\tau} \right) - \phi \left( S_{\tau},    M^{n, p}_{\tau} \right) \right|.
\end{eqnarray*}
By symmetry and $\mathrm{(A2)}$, we get
\begin{eqnarray*}
 \left|	\ \sup_{\tau} \E \left[ \phi \left( S_{\tau}, M_{\tau} \right) \right] - \sup_{\tau} \E \left[ \phi \left( S_{\tau},    M^{n, p}_{\tau} \right) \right] \right|
		& \leq & \sup_{\tau} \E \left| M_{\tau} -   M^{n, p}_{\tau} \right| \\
		& \leq  &  \mathbb E \left[\sup_{0\leq t\leq T} |M_t - M^{n,p}_t| \right]
\end{eqnarray*}
and the result follows from Proposition \ref{prop-erreur-mm}.
\end{proof}

%%%%%%%%%%%%%%%%%%%%%%%%%%%%%%%%%%%%%%%%%%%%%%%%%%%%%%%%
%%%%%%%%%%%%%%%%%%%%%%%%%%%%%%%%%%%%%%%%%%%%%%%%%%%%%%%%
\paragraph{Uniformly-weighted moving average}
\label{subsec-cas-uniforme}
The uniform weighting measure
\begin{equation}
\mu(dx) = h(x)dx = \frac{1}{\delta} \Ind_{[0, \delta]}dx
\label{heavyside}
\end{equation}
satisfies the assumptions of Proposition \ref{prop-erreur-mm}.
In particular, $H (x) = \frac{1}{\delta} \left( \delta - x
\right)^{+}$.  From Lemma \ref{deriv.lm}, the Laguerre coefficients $A^{\delta, p}_k = \langle H,
L^p_k \rangle$ are related to the Laguerre coefficients of $h$,
$c^{\delta,p}_k = \left\langle h, L^{p}_k \right\rangle_2$, via
\begin{eqnarray}
	A^{\delta, p}_k & = & (-1)^k \frac{\sqrt{2p}}{p} - \frac{1}{p} c^{\delta,p}_k - \frac{2}{p} \sum_{i = 0}^{k-1} (-1)^{k-i} c^{\delta, p}_i,
\label{coef-A-H}
\end{eqnarray}
and the coefficients of $h$ can be computed from the values of
Laguerre polynomials:
\begin{align}
c^{\delta,p}_n = \frac{\sqrt{2 p}}{\delta p} \left[ (1-e^{-p\delta} P_n(2p\delta))
+ 2 \sum_{k=1}^n (-1)^k (1-e^{-p\delta}
P_{n-k}(2p\delta)) \right]. \label{formule-coef-fourier}
\end{align}

Given the length of the averaging window $\delta > 0$ and an order $n
\geq 1$ of approximation (number of Laguerre functions), we determine the optimal scale parameter $p_{\opt}(\delta, n)$ as
\begin{align}
p_{\opt}(\delta, n) = \argmin_{p > 0} \left\| H - H^p_n \right\|_2  =
\argmin_{p > 0} \left\{\frac{\delta}{3} - \sum_{k = 0}^{n-1} \left|
    A^{\delta, p}_k \right|^2\right\}.
  \label{error-H}
\end{align}
Finding an explicit formula to $p_{\opt}(\delta, n)$ does not seem to be
possible, but finding a numerical solution is easy using the explicit expressions \eqref{coef-A-H}, \eqref{formule-coef-fourier} and \eqref{error-H}.
%, Lemma \ref{formule-explicite-coef-h}.
In addition, once $p_{\opt}$ is computed for a couple $(1, n)$, the scaling property of Laguerre functions \eqref{def-laguerre-parametre} gives the value of $p_{\opt} (\delta, n)$, for any $\delta > 0$:
$$
p_{\opt} (\delta, n) = \frac{p_{\opt} (1, n)}{\delta}.
$$
%Besides, this same scaling property implies that the $\LL^2$-error between $H$ and $H^p_n$ is \textit{scale invariant} so that the behaviour of $\left\| H - H^{p_{\opt}}_n \right\|_2$ shown in Figure \ref{approx-H-opti} holds for any $\delta > 0$.

Table \ref{tab-opt-scale-parameter} gives the values for $p_{\opt} (1, n)$ for the
first $10$ values of $n$ computed with an accuracy of $10^{-3}$.
Figure \ref{approx-H-opti} (left graph) illustrates the approximation
of $H$ by the truncated Laguerre expansion $H^{p_{\opt}(n)}_n$ for $n
= 1, 3, 7$ Laguerre basis functions (with $\delta=1$). The corresponding error $|| H -
H^{p_{\opt}(n)}_n ||_2$ as a function of $n$ is shown in the right graph. The error is less than $5 \%$ already with $n = 3$. A simple least squares estimation by a power function gives a behavior in $\Ocurs( n^{-1.06} )$. \\

\UnTableau{
\begin{tabular}{|c|cccccccccc|}
	\hline
	$n$ & 1 & 2 & 3 & 4 & 5 & 6 & 7 & 8 & 9 & 10 \\
	\hline
	$p_{\opt}(1, n)$ & 2.149 & 4.072 & 6.002 & 4.234 & 5.828 & 7.473 & 9.155 & 10.866 & 9.153 & 10.726 \\
	\hline
\end{tabular}
}{Optimal scaling parameters for approximating $H (x) = \frac{1}{\delta} \left( \delta - x \right)^{+}$.}{tab-opt-scale-parameter}{0}{0.2}
%
%\begin{figure}
%\centerline{\includegraphics[width=0.55\textwidth]{graphes/approx_H_n1-3-7.pdf}\hspace*{-0.8cm}\includegraphics[width=0.55\textwidth]{graphes/CV_minimal_error_L2H.pdf}}
%\caption{Laguerre approximation of the function $H (x) = \frac{1}{\delta} \left( \delta - x
%\right)^{+}$.}
%\label{approx-H-opti}
%\end{figure}
%
% SANS LES CAPTIONS DES 2 GRAPHES et ZOOM
\DeuxFiguresACote{0.2}{8}{0}{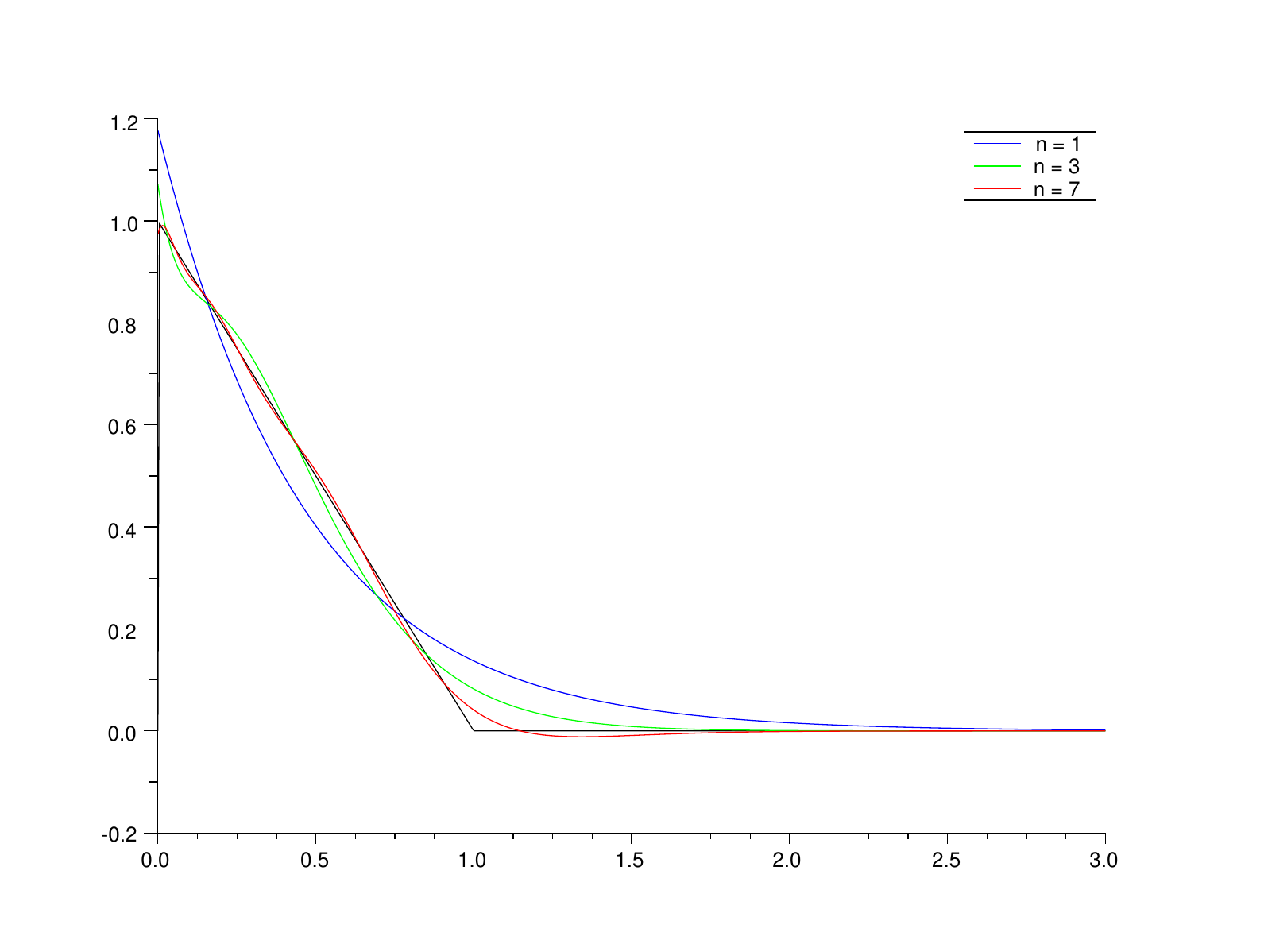}{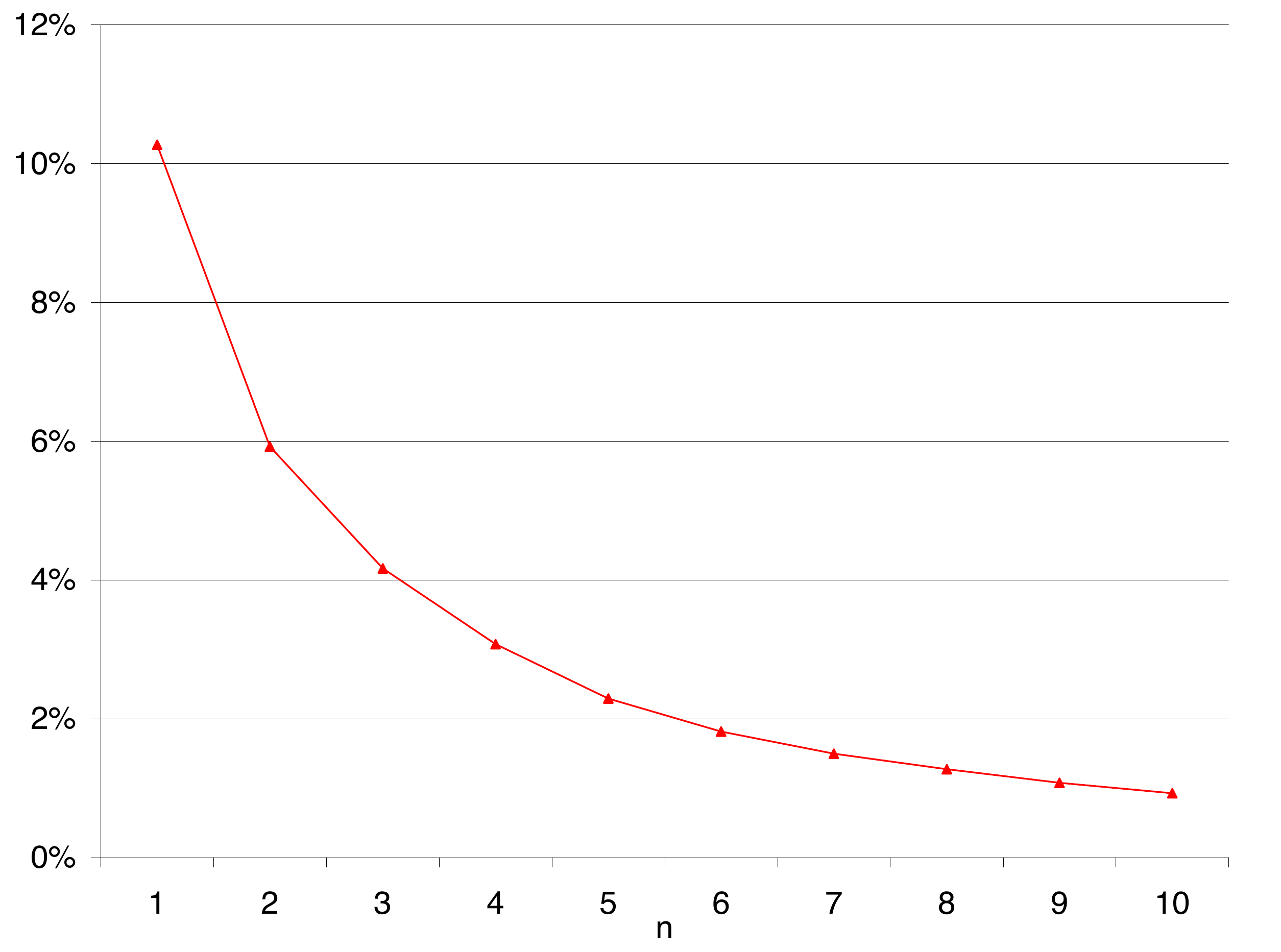}{Left:
  Laguerre approximation of the function $H (x) = \frac{1}{\delta}
  \left( \delta - x \right)^{+}$. Right: $L^2$ error of the approximation.}{approx-H-opti}{-0.1}{-0.2}

%Figure \ref{CV-approx-H-opti} shows the corresponding error $|| H - H^{p_{\opt}(n)}_n ||_2$ as a function of $n$. The error is less than $5 \%$ already with $n = 3$. A simple Least Squares estimation by a power function gives a behavior in $\Ocurs( n^{-1.06} )$.
%The optimal scaling of the Laguerre functions leads thus to a \textit{convergence rate improvement}, in comparison to the convergence rate in $\Ocurs(n^{-\frac{3}{4}})$ (uniform in $p$)  given by Lemma \ref{coefs.lm}.
%This convergence rate is improved by optimal scaling of the Laguerre functions, in comparison to

Proposition \ref{prop-erreur-mm} also applies when the moving average
is delayed by a fixed time lag $l \geq 0$:
\begin{align}
X_t = \frac{1}{\delta} \int_{t - l - \delta}^{t - l} S_{u} du,\ \forall t \geq \delta + l.
\label{rem-ma-option-with-lag}
\end{align}
In this case, the weighting measure is $\mu(dx) = \frac{1}{\delta}
\Ind_{[l, l + \delta]} dx$ and $H(x) = \frac{1}{\delta}\left\{(\delta
  + l - x)^+ - (l-x)^+ \right\}$.
Figure \ref{approx-H-opti-lag} shows the approximation of $H$ by
$H^{p_{\opt}(n)}_n$ for $n = 1, 3, 5, 7$ Laguerre basis functions with
$\delta = 1$ and $l = 0.5$ as well as the $\LL^2$-error made as a function of $n$
(in the same way as above, we numerically compute $p_{\opt}(n)$ for minimizing the $\LL^2$-error made by Laguerre approximation).
In comparison to the previous case, it appears that
the number of Laguerre basis functions necessary to approximate $H$ is
greater for an equivalent accuracy of the approximation: the error is
less than $5 \%$ from $n = 5$. This is due to the fact that the
density of the weighting measure has two points of discontinuity. \\

%\begin{figure}
%\centerline{\includegraphics[width=0.55\textwidth]{graphes/approx_H_lag_n1-3-5-7.pdf}\hspace*{-0.8cm}\includegraphics[width=0.55\textwidth]{graphes/CV_minimal_error_L2H_lag.pdf}}
%\caption{Laguerre approximation of the function $H(x) = \frac{1}{\delta}\left\{(\delta
%  + l - x)^+ - (l-x)^+ \right\}$.
%}
%\label{approx-H-opti-lag}
%\end{figure}

% SANS LES CAPTIONS DES 2 GRAPHES et ZOOM
\DeuxFiguresACote{0.2}{8}{0}{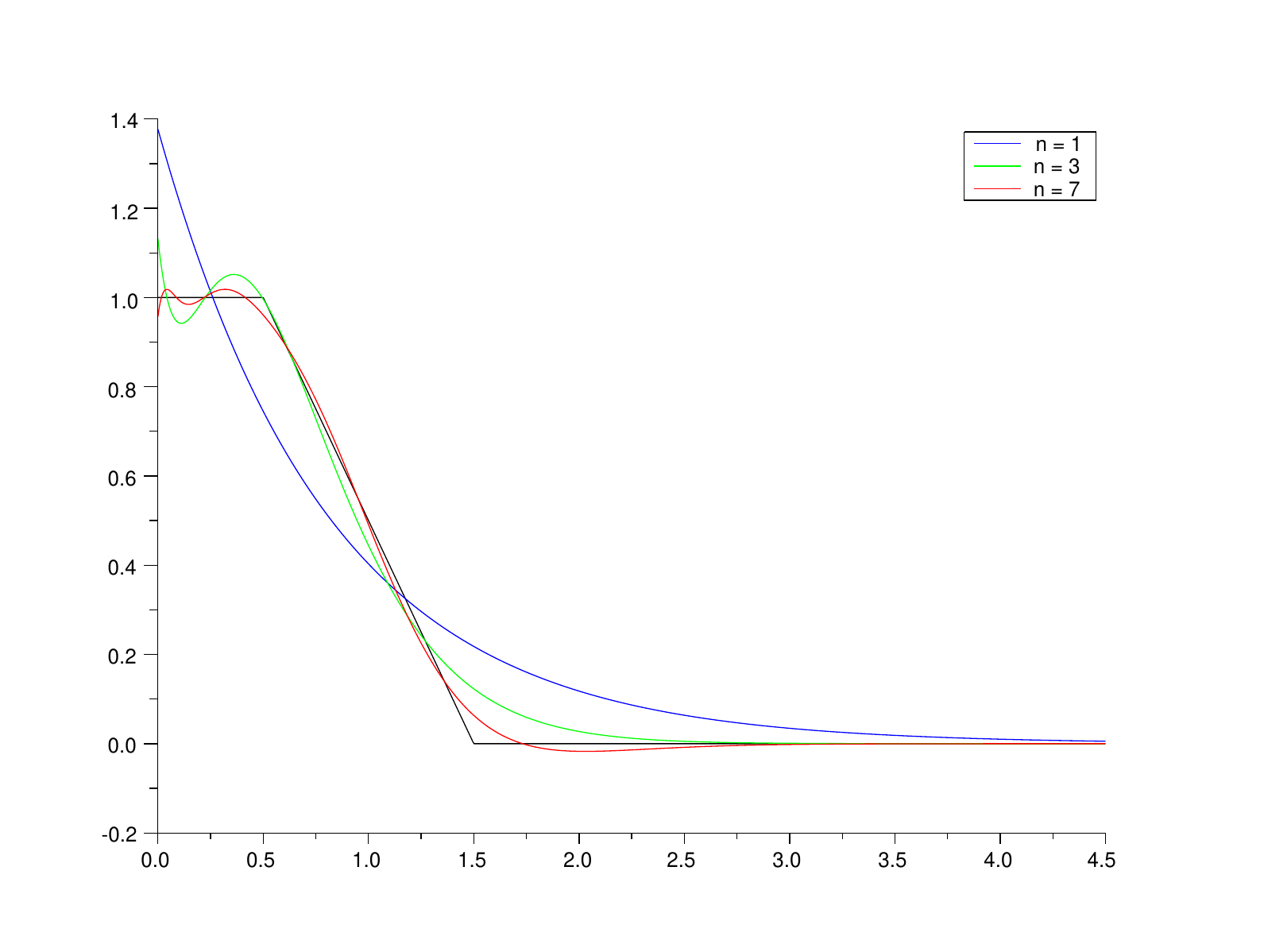}{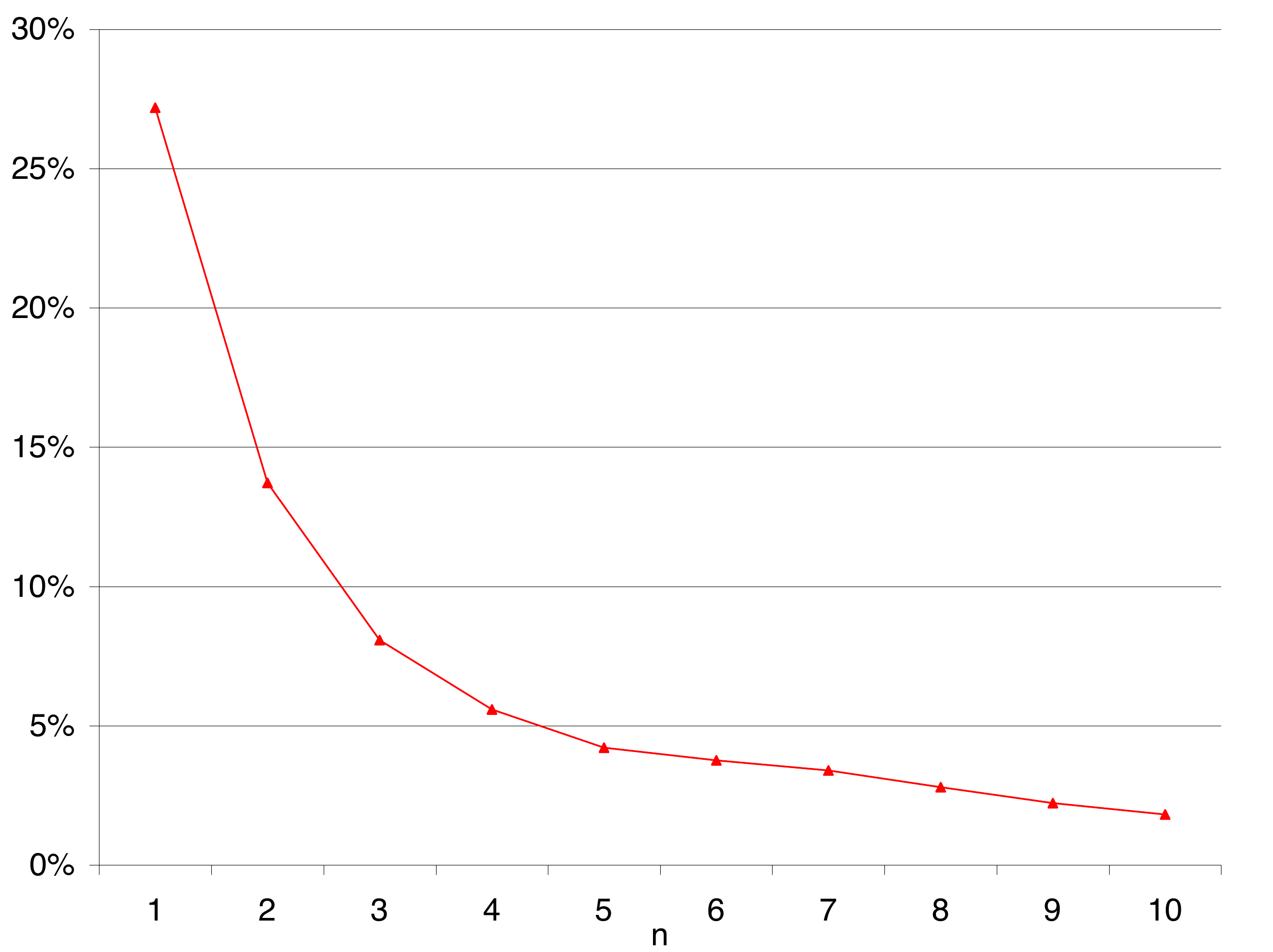}{Left:
  Laguerre approximation of the function $H(x) =
  \frac{1}{\delta}\left\{(\delta  + l - x)^+ - (l-x)^+
  \right\}$. Right: $L^2$ error of the approximation.}{approx-H-opti-lag}{-0.1}{-0.2}

%%%%%%%%%%%%%%%%%%%%%%%%%%%%%%%%%%%%%%%%%%%%%%%%%%%%%%%%
%
%%%%%%%%%%%%%%%%%%%%%%%%%%%%%%%%%%%%%%%%%%%%%%%%%%%%%%%%
\section{A Monte Carlo-based numerical method}
\label{sec-method-num}

%Numerically, this result means that we are reduced to

In this section we present a numerical method for computing the solution to the approximate problem \eqref{def-approx-pb}, which is a $(n+1)$-dimensional optimal stopping time problem.
For the sake of simplicity, we restrict ourselves to uniformly
weighted moving averages. Since the dimension of the problem may be
high, we use a Monte Carlo technique. Our numerical approach corresponds to the one from Longstaff and Schwartz \cite{LS01} and the computation of conditional expectations is done with a regression based approach. In particular, we shall use the technique of adaptative local basis proposed by Bouchard and Warin \cite{BW10}.

\paragraph{Forward simulation in discrete time}
We compute the price of the \textit{discrete time version} of the American option \eqref{def-approx-pb} in which the moving average $X$ has been replaced by its approximation $M^{n, p_{\opt}}$ defined in \eqref{def-approx-X-sum}
and the exercise is possible on an equidistant discrete time grid with $N \geq 1$ time steps $\Delta t = \frac{T}{N}$:
$
\pi = \left\{t_0 = 0, t_1, \ldots, t_N = T\right\}
$. We assume that there are exactly $N_{\delta} \geq 1$ time steps
within the averaging window of length $\delta$: $N_{\delta} =
\frac{\delta}{\Delta t} = \frac{\delta}{T} N$, and that the spot price
$S$ can be simulated on $\pi$ either exactly or using the Euler scheme.
%Recall that the Markovian state vector of this problem is:
%$$
%\left( S, X^{p_{\opt}, 0}, X^{p_{\opt}, 1}, \ldots, X^{p_{\opt}, n-1} \right)
%$$
%
We denote this simulated discrete time price by
$$
\left\{S^{\pi}_{t_0} = S_0, S^{\pi}_{t_1}, \ldots, S^{\pi}_{t_N}\right\},
$$
extend this definition to $[0, T]$ by
\begin{equation}
S^{\pi}_t = S^{\pi}_{t_{i+1}}, \ \forall t \in (t_{i}, t_{i+1}]
\label{def-ext-Spi}
\end{equation}
and shall also apply convention \eqref{convention-S} to $S^{\pi}$.
We define the discrete time version $X^{\pi}$ of the moving average process $X$ by
\begin{equation}
X^{\pi}_{t_i} = \frac{1}{\delta} \int_{t_i - \delta}^{t_i} S^{\pi}_t dt
 =  \frac{1}{N_{\delta}} \sum^{i}_{j = i - N_{\delta} + 1}  S^{\pi}_{t_j}, \ \forall t_i \in \pi.
\label{def-Xpi}
\end{equation}
%which corresponds to a simple arithmetic average including the current time. This definition is consistent with the definition of $(S^{\pi}_t)_{0 \leq t \leq T}$ in \eqref{def-ext-Spi} because:
%The processes $(X^{p, k})_{k \geq 0}$ can be computed exactly on any time grid, as soon as the discrete time version of the price process $S$ is available.
%%%%%%%%%%%%%%%%%%%%%%%%%%%%%%%%%%%%%%%%%
Similarly, the discrete time versions of the Laguerre processes are
defined by
\begin{equation*}
X^{p, k, \pi}_{t_i} = \int_{-\infty}^{t_i} L^p_k(t_i-v)S^\pi_t dt=  \sum_{j = 1}^{i} \left( S^{\pi}_{t_j} - S^{\pi}_{t_{j-1}} \right) (i-j+1) \Delta t \ c^{(i-j+1) \Delta t, p}_{k}
		+ S_0 (-1)^k  \frac{\sqrt{2p}}{p}, \ \forall t_i \in \pi.
\end{equation*}

%%%%%%%%%%%%%%%%%%%%%%%%%%%%%%%%%%%%%%%%%%%%%%%%%%%%%%%%
%%%%%%%%%%%%%%%%%%%%%%%%%%%%%%%%%%%%%%%%%%%%%%%%%%%%%%%%
\paragraph{Backward resolution of the optimal stopping time problem}

The resolution is based on the well-known backward American dynamic programming principle. We adopt a Longstaff and Schwartz-style approach
which consists in estimating the optimal exercise time (or
equivalently the optimal cashflows generated by the optimal exercise
rule) instead of focusing on the computation of the option value
processes (as for example in Tsitsiklis and Van Roy
\cite{TvR01}). Throughout the paper, the approach presented below will be called (Lag-LS).
The optimal payoffs are evaluated using the \textit{approximate value} of moving average $X^{\pi}$ derived from \eqref{def-approx-X-sum}:
\begin{equation}
M^{n, p_{\opt}, \pi}_{t_i} = (H(0)-H^{p_{\opt}}_n(0)) S^{\pi}_{t_i} + \sum_{k=0}^{n-1} a^{p_{\opt}}_k X^{p_{\opt}, k, \pi}_{t_i}, \forall t_i \in \pi.
\label{def-Mpi}
\end{equation}
%in which we recall that the coefficients $a^{p}_k$ can be explicitly computed from their definition in \eqref{def-approx-I}, expressions \eqref{coef-A-H} and \eqref{formule-coef-fourier}. \\

Denote by %$(V^{\pi}_{t_i})_{t_i \in \pi}$ the discrete time price process of the American option and
$(\tau^{\pi}_{i})_{i = N_{\delta}, \ldots, N}$ the sequence of discretized optimal exercise times: $\tau^{\pi}_i$ is the optimal exercise time after $t_i \in \pi$.
The backward algorithm works as follows:
\begin{enumerate}
	\item Initialization: $\tau^{\pi}_{N} = T$
	\item Backward induction for $i = N-1, \ldots, N_{\delta}$:
	\begin{eqnarray*}
		\tau^{\pi}_{i} = t_i \Ind_{A_{i}} + \tau^{\pi}_{i+1} \Ind_{\complement A_{i}}
		\text{ with } A_{i} = \left\{ \phi \left( S^{\pi}_{t_i}, M^{n, p_{\opt}, \pi}_{t_i} \right) \geq \E_{t_{i}} \left[ \phi \left( S^{\pi}_{\tau^{\pi}_{i+1}}, M^{n, p_{\opt}, \pi}_{\tau^{\pi}_{i+1}} \right) \right] \right\}
	\end{eqnarray*}
	\item Estimation of the option price at time $0$:
	\begin{eqnarray*}
		V^{\pi}_0 = \E \left[ \phi \left( S^{\pi}_{\tau^{\pi}_{N_{\delta}}}, M^{n, p_{\opt}, \pi}_{\tau^{\pi}_{N_{\delta}}} \right) \right]
	\end{eqnarray*}
\end{enumerate}
in which:
$$
\E_{t_{i}} \left[ \cdot \right] = \E \left[ \cdot | \left( S^{\pi}_{t_i}, X^{p_{\opt}, 0, \pi}_{t_i}, \ldots, X^{p_{\opt}, n-1, \pi}_{t_i} \right) \right].
$$
Estimators of the conditional expectations are constructed with a Monte-Carlo based technique. It consists in using $M \geq 1$ simulated paths on $\pi$ of the $(n+1)$-dimensional state process:
$$
\left( S^{\pi, (m)}, X^{p_{\opt}, 0, \pi, (m)}, \ldots, X^{p_{\opt},
    n-1, \pi, (m)} \right), \ m=1,\dots,M.
$$
The corresponding paths of the approximate moving average are denoted by:
$$
M^{n, p_{\opt}, \pi, (m)}, \ m=1,\dots,M.
$$
%On any path $m \leq M$, the useful values of the moving average on grid $\pi$ are:
%\begin{equation*}
%X^{\pi, (m)}_{t_i} = \frac{1}{N_{\delta}} \sum^{i}_{j = i - N_{\delta} + 1}  S^{\pi, (m)}_{t_j}, \ \forall i = N_{\delta}, \ldots, N
%\end{equation*}
Conditional expectations estimators $\E^{M}_{t_i}$ are then computed by regression on local basis functions (see the precise description of the procedure in Bouchard and Warin \cite{BW10}). We shall denote by $\left(b^{S}, b^{X}_0, \ldots, b^{X}_{n-1} \right)$ the numbers of basis functions used in each direction of the state variable: $b^{S}$ for $S^{\pi}$, $b^{X}_0$ for $X^{p_{\opt}, 0, \pi}$, $b^{X}_1$ for $X^{p_{\opt}, 1, \pi}$, etc. The Monte-Carlo based backward procedure becomes thus:
\begin{enumerate}
	\item Initialization: $\tau^{\pi, (m)}_{N} = T$, $m=1,\dots,M$
	\item Backward induction for $i = N-1, \ldots, N_{\delta}$, $m=1,\dots,M$:
	\begin{eqnarray*}
	\begin{cases}
		\tau^{\pi, (m)}_{i} = t_i \Ind_{A^{(m)}_{i}} + \tau^{\pi, (m)}_{i+1} \Ind_{\complement A^{(m)}_{i}} \\
		A^{(m)}_{i} = \left\{ \phi \left( S^{\pi, (m)}_{t_i}, M^{n, p_{\opt}, \pi, (m)}_{t_i} \right) \geq \E^{M}_{t_i} \left[ \phi \left( S^{\pi}_{\tau^{\pi}_{i+1}}, M^{n, p_{\opt}, \pi}_{\tau^{\pi}_{i+1}} \right) \right] \right\}
		%(S^{\pi, (m)}_{t_i}, X^{p_{\opt}, 0, \pi, (m)}_{t_i}, \ldots, X^{p_{\opt}, n-1, \pi, (m)}_{t_i})
	\end{cases}
	\end{eqnarray*}
	\item Estimation of the option price at time $0$:
	\[ \begin{array}{c} %\E^{M} \left[ \phi \left( S^{\pi}_{\tau^{\pi}_{N_{\delta}}}, X^{\pi}_{\tau^{\pi}_{N_{\delta}}} \right) \right] =
	V^{\pi}_0 = \frac{1}{M} \sum_{m=1}^{M}  \phi \left( S^{\pi, (m)}_{\tau^{\pi, (m)}_{N_{\delta}}}, M^{n, p_{\opt}, \pi, (m)}_{\tau^{\pi, (m)}_{N_{\delta}}} \right)
\end{array} \]
\end{enumerate}

\begin{rem} We will use a numerical improvement to this standard backward induction algorithm, which might seem rather natural for practitioners. % associated to American-style options
It consists in evaluating the optimal payoffs using the \textit{exact value} \eqref{def-Xpi} of the moving average.
In particular, the optimal stopping frontier becomes:
$$
 A^{*}_{i} = \left\{ \phi \left( S^{\pi}_{t_i}, X^{\pi}_{t_i} \right) \geq \E_{t_{i}} \left[ \phi \left( S^{\pi}_{\tau^{\pi}_{i+1}}, X^{\pi}_{\tau^{\pi}_{i+1}} \right) \right] \right\}.
$$
This improved method will be called (Lag-LS*) and unlike (Lag-LS) will
exhibit a monotone convergence as $n$ goes to infinity.
\end{rem}

\paragraph{"Non Markovian'' approximation for moving average options}
Motivated by a reduction of dimensionality,
the numerical method that is most often used in practice to value
moving average options consists in computing the conditional
expectations in the Longstaff-Schwartz algorithm using only the explanatory variables $(S, X)$: namely, the price and the moving average appearing in the option payoff.
The resulting exercise time is thus suboptimal, but the approximate
option price is often close to the true price.
To assess the improvement offered by our method, we systematically
compare our approximation to this suboptimal approximate price, also
computed using a Longstaff and Schwartz approach and referred to as
(NM-LS). 

Let $(\theta^{\pi}_{i})_{i = N_{\delta}, \ldots, N}$ denote the
discrete time sequence of the estimated optimal exercise times
($\theta^{\pi}_i$ being the optimal exercise time after $t_i \in
\pi$).  (NM-LS) algorithm works as follows:
\begin{enumerate}
	\item Initialization: $\theta^{\pi}_{N} = T$
	\item Backward induction for $i = N-1, \ldots, N_{\delta}$:
	\begin{eqnarray*}
		\theta^{\pi}_{i} = t_i \Ind_{A_{i}} + \theta^{\pi}_{i+1} \Ind_{\complement A_{i}}
		\text{ with } A_{i} = \left\{ \phi \left( S^{\pi}_{t_i}, X^{\pi}_{t_i} \right) \geq \E \left[ \phi \left( S^{\pi}_{\theta^{\pi}_{i+1}}, X^{\pi}_{\theta^{\pi}_{i+1}} \right)  | \left( S^{\pi}_{t_i}, X^{\pi}_{t_i} \right)\right] \right\}
	\end{eqnarray*}
	\item Estimation of the option price at time $0$:
	\begin{eqnarray*}
		U^{\pi}_0 = \E \left[ \phi \left( S^{\pi}_{\theta^{\pi}_{N_{\delta}}}, X^{\pi}_{\theta^{\pi}_{N_{\delta}}} \right) \right]
	\end{eqnarray*}
\end{enumerate}
Similarly to other methods, the conditional expectations are computed with the adaptative local basis regression-based technique from \cite{BW10}. The numbers of basis functions used in each direction will be denoted by $b^{S}$ for $S^{\pi}$ and $b^{X}$ for $X^{\pi}$.

%%%%%%%%%%%%%%%%%%%%%%%%%%%%%%%%%%%%%%%%%%%%%%%%%%%%%%%%
%
%%%%%%%%%%%%%%%%%%%%%%%%%%%%%%%%%%%%%%%%%%%%%%%%%%%%%%%%
\section{Numerical examples}
\label{sec-part-num}

For our examples, we use the single-asset Black and Scholes
framework. We study standard moving average options for different
values of the averaging window $\delta$ as well as moving average
options with delay \eqref{rem-ma-option-with-lag}.
%In high dimension (when $\delta$ is large), our results are compared to the one obtained by the non Markovian approximate method which is most often used by practitioners (see Paragraph \ref{subsec-NM-approx}).

With the same notations as in Section \ref{sec-method-num}, recall that the dimension of the discrete time version of moving average option pricing problem is equal to $N_{\delta}$ with a Markovian state:
$$
\left( S^{\pi}_{t_i}, S^{\pi}_{t_{i-1}}, \ldots, S^{\pi}_{t_{i - N_{\delta + 1}}} \right), \forall t_i \in \pi, t_i \geq t_{N_{\delta}}.
$$
We use the standard Longstaff and Schwartz algorithm for such a Bermudan option in dimension $N_{\delta}$ as the benchmark method.
This method will be called (M-LS) and our Monte Carlo regression based
approach (see more details in \cite{BW10}) allows to deal with cases
up to dimension $8$. For applications in which $N_{\delta}$ is larger,
this method becomes computationally unfeasible.

%MB
We provide at the end of this section a numerical comparison between the convergence rate of our Laguerre-based approximation and (M-LS) with respect to the state dimension.

\paragraph{Moving average options: %in the Black and Scholes framework:
  benchmark prices}
Consider a standard moving average American option  with value at time $0$:
$$
\sup_{\tau \in \Tcurs_{[\delta, T]}} \E \left[ e^{-r \tau} \phi \left(
    S_{\tau}, X_{\tau}
  \right) \right],\quad X_{\tau} = \frac{1}{\delta} \int^{\tau}_{\tau-\delta} S_u du
$$
where the asset price $S$ is assumed to follow the risk-neutral Black and Scholes dynamics:
\begin{eqnarray*}
dS_t = S_t \left( r dt + \sigma dW_t \right), \ S_0 = s
\end{eqnarray*}
and $W$ is a standard Brownian motion. We shall consider call options
with pay-off $\phi(s, x) = (s - x)^{+}$.
% as well as the linear payoff $\phi(s, x) = s - x$.
Unless specified otherwise, the following
parameters are used below:
\begin{center}
\begin{tabular}{l|l}
	%Option type & Moving window Call Asian, with floating strike \\
	Maturity & $T = 0.2$ \\
%	Payoff & $\phi(s, x) = \left( x - s \right)^{+}$ \\
%	Length of the moving window & $\delta = 0.04$ \\
	Risk free interest & $r = 5 \%$ \\
	Volatility & $\sigma = 30 \%$ \\
	Initial spot value & $s = 100$ \\
\end{tabular}
\end{center}
and we consider a Bermudan option with exercise possible every day (when $T = 0.2$, the time interval $[0, T]$ is divided into $N = 50$ time steps).

Table \ref{tab-comp-MLS-NMLS} shows the prices of moving average call options computed by (NM-LS) and (M-LS) for various averaging periods $\delta$, with $M = 10$ million of Monte Carlo paths and $b^S = b^X = 2$. The prices are averages over 5 valuations and the relative standard deviation is given in brackets. 
For reasonable volatility coefficients of the underlying price process
and relatively small averaging window $\delta$, (NM-LS) seems to
provide a very good approximation (from below) to moving average
options prices. This justifies the approximation made by
practitioners and (among others) by Broadie and Cao \cite{BC07}. \\

\UnTableau{
\begin{tabular}{|c|c|c|}
	\hline
		%& \multicolumn{2}{|c|}{NM-LS} & \multicolumn{2}{|c|}{M-LS} \\
	$N_{\delta}$ & (NM-LS) & (M-LS) \\
	\hline
	 2 & 1.890	(0.011 \%) & 1.890	(0.011 \%) \\
	 3 & 2.684	(0.011 \%) & 2.685	(0.010 \%) \\
	 4 & 3.183	(0.018 \%) & 3.186	(0.012 \%) \\
	 5 & 3.526 	(0.016 \%) & 3.531	(0.007 \%) \\
	 6 & 3.773 	(0.016 \%) & 3.780	(0.013 \%) \\
	 7 & 3.955	(0.011 \%) & 3.964	(0.215 \%) \\
	 8 & 4.092	(0.015 \%) & 4.103	(0.316 \%) \\
	 9 & 4.193	(0.016 \%) &  \\
	 10 & 4.268	(0.019 \%) &  \\
	\hline
	\end{tabular}
}{Moving average options pricing with (NM-LS) and (M-LS).}{tab-comp-MLS-NMLS}{0}{0.2}
%In high dimension (greater than $8$), we will use this non Markovian approximation (NM-LS) as a benchmark.

%\begin{rem}
%The problem dimension increases much more when pricing delayed moving average options as presented in Remark \ref{rem-ma-option-with-lag}: the dimension is equal to the number of time step within the averaging window $N_\delta$ plus the number of time steps within the time lag $l$ (it will be denoted by $N_{l} = \frac{l}{\Delta t} = \frac{l}{T} N$). However, in this last case, we can expect that (NM-LS) would be a worse approximation.
%\end{rem}

% SANS CAPTION et ZOOM
\UneFigure{15}{13}{0}{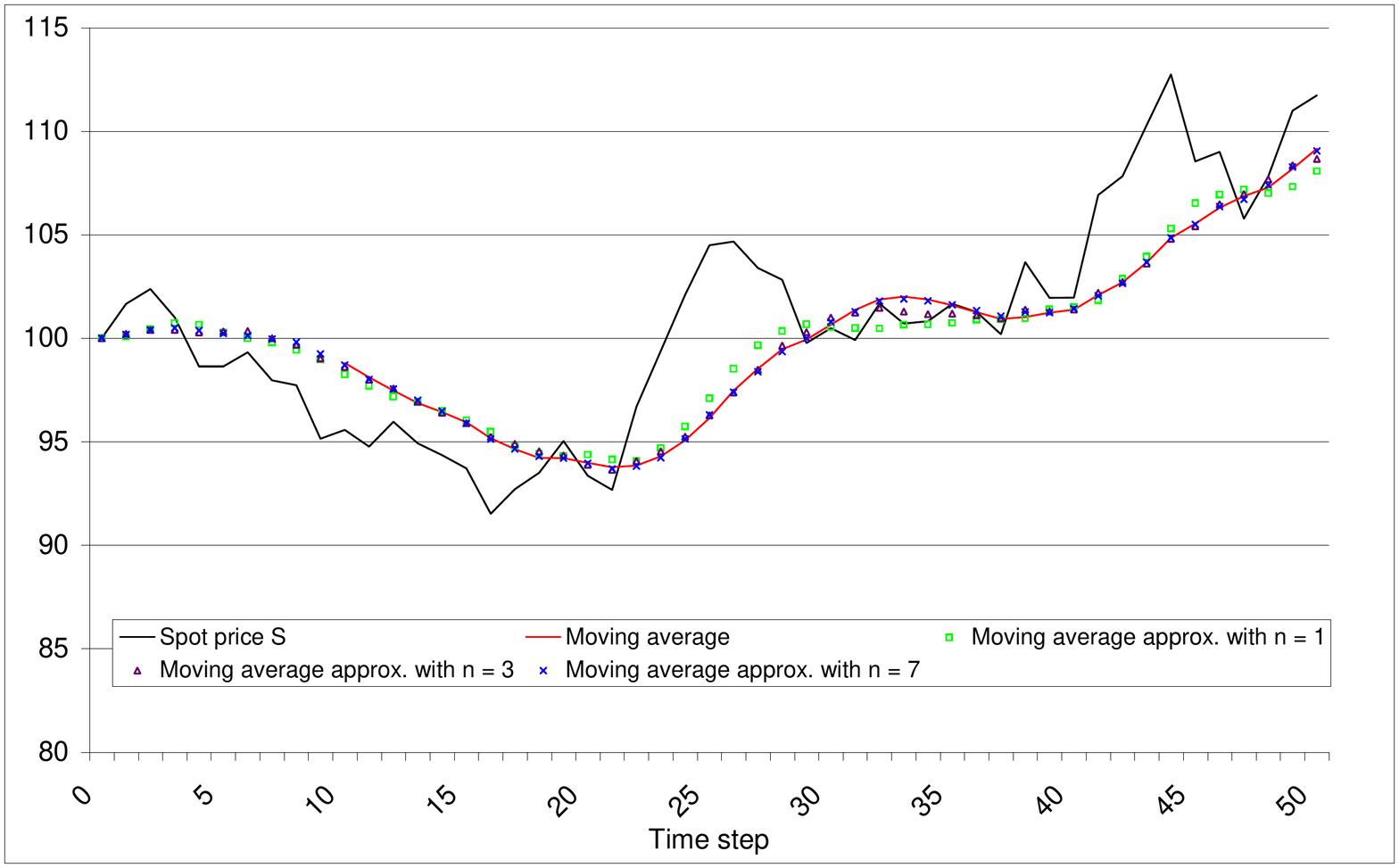}{Simulated trajectory of the moving average process and its
Laguerre approximations.}{simu-opti-MM}{-1}{-0.8}

%
%%%%%%%%%%%%%%%%%%%%%%%%%%%%%%%%%%%%%%%%%%%%%%%%%%%%%%%%
%%%%%%%%%%%%%%%%%%%%%%%%%%%%%%%%%%%%%%%%%%%%%%%%%%%%%%%%
\paragraph{Moving average American options: %in the Black and Scholes model:
		Laguerre approximation}
\label{application-BS-simple}

Figure \ref{simu-opti-MM}
shows a simulated trajectory of the underlying price $S^{\pi}$, its moving average
$X^{\pi}$ with $\delta = 0.04$ and the corresponding Laguerre-based moving average
approximation $M^{n, p_{\opt}, \pi}$ with $n = 1, 3$ and $7$ Laguerre
basis functions. Already for $n \geq 3$ Laguerre basis functions,
$M^{n, p_{\opt}, \pi}$ accurately mimics the exact moving average
dynamics of $X^{\pi}$ and this approximation seems to be almost exact
when $n = 7$.

Table \ref{tab-comp-LS-Lag-opti} reports the prices of moving average call options computed using the Laguerre-based method presented in Section \ref{sec-method-num} (Lag-LS) and its improved version (Lag-LS*) (with the same parameters as above, in particular $\delta = 0.04$). The price values are means over 5 valuations, the relative standard deviation is given in brackets and we used $M = 5$ million Monte Carlo paths for $n = 1, \ldots, 3$ Laguerre functions and $M = 10$ million Monte Carlo paths for $n = 4, \ldots, 7$ Laguerre functions, with $b^S = 4$ and $b^X_k = 1, \forall k \geq 0$.
With $M = 10$ million Monte Carlo paths, $b^S = 4$ and $b^X = 1$, (NM-LS) gives an option value equal to $4.268$. \\

\UnTableau{
\begin{tabular}{|c|c|c|}
	\hline
		%& \multicolumn{2}{|c|}{NM-LS} & \multicolumn{2}{|c|}{M-LS} \\
	$n$ & (Lag-LS*) & (Lag-LS) \\
	\hline
	 1 & 4.266	(0.020 \%) & 4.092	(0.017 \%) \\
	 2 & 4.273	(0.022 \%) & 4.302	(0.019 \%) \\
	 3 & 4.276	(0.023 \%) & 4.182	(0.018 \%) \\
	 4 & 4.276	(0.022 \%) & 4.227	(0.020 \%) \\
	 5 & 4.277	(0.023 \%) & 4.275	(0.020 \%) \\
	 6 & 4.277	(0.024 \%) & 4.287	(0.022 \%) \\
	 7 & 4.277	(0.024 \%) & 4.258	(0.022 \%) \\
	\hline
	\end{tabular}
}{Moving average options pricing with (Lag-LS) and (Lag-LS*).}{tab-comp-LS-Lag-opti}{0}{0.2}
\begin{rem}
%Number of function basis used in the regression based approach to computer the conditional expectations estimators
% (large $\delta$ with respect to $T$)
When the averaging window is large, the variance of the Laguerre states $(X^{p, k}_{k})_{k \geq 0}$ is small, and at least much smaller than the variance of the price $S$.
In consequence, increasing the numbers $b^X_k$ of basis functions in
the directions of these states does not have a strong impact on the conditional expectation estimators and the resulting option price. On the contrary, the number $b^S$ of basis functions in the direction of the spot price $S$ should be sufficiently large.
\end{rem}

Whereas (Lag-LS) oscillates as $n$ increases (this is due to the non
monotone approximation of the moving average $X$ by $M^{n,
  p_{\opt}}$), (Lag-LS*) shows a monotone convergence when increasing
$n$, as shown in Figure \ref{valo-LS-Lag-linear-delta10}. The limiting
value (almost $4.277$) is around $0.2 \%$ above the value computed by
the practitioner's approximation (NM-LS). \\

%Corresponding graphic with a benchmark value computed by (NM-LS) with $M = 10$ million of MC paths, $b^S = 4$, $b^X = 1$
%
%\begin{figure}
%\centerline{\includegraphics[width=0.7\textwidth]{graphes/LS-Lag-opti_linear_delta10.pdf}}
%\caption{Convergence of the improved Laguerre-based approximation.}
%\label{valo-LS-Lag-linear-delta10}
%\end{figure}
%
% SANS CAPTION et ZOOM pour CALL
\UneFigure{11}{11}{0}{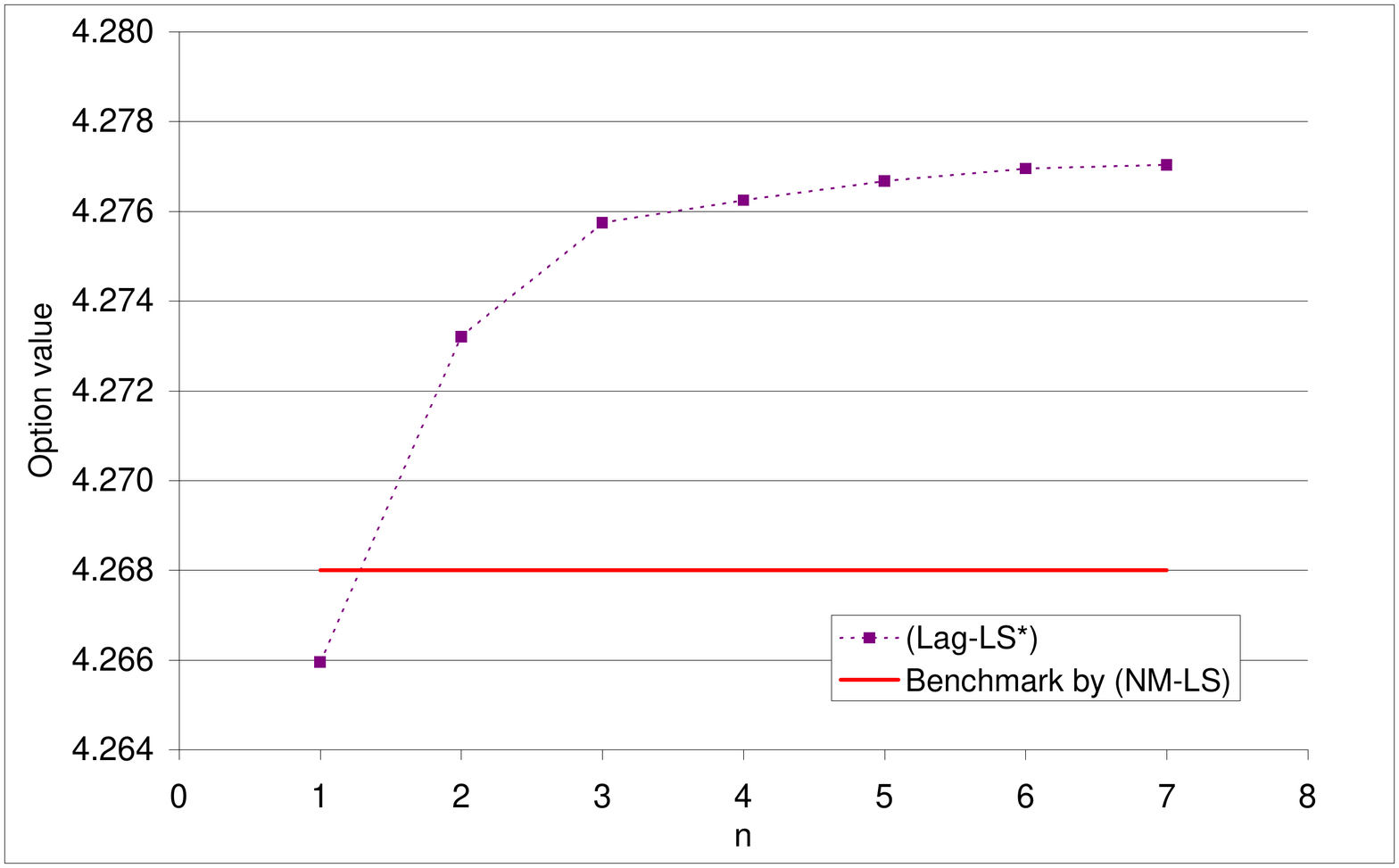}{Convergence of the improved Laguerre-based approximation.}{valo-LS-Lag-linear-delta10}{-1}{-0.8}

%[MB: peut-être pas nécessaire de mettre ce tableau \ref{tab-comp-NMLS-LSLag} ?]
%Table \ref{tab-comp-NMLS-LSLag} reports the prices of moving average call options computed by (Lag-LS*) for various averaging period $\delta$, with $M = 10$ million Monte Carlo paths, $b^S = 3$ and $b^X_k = 1, \forall k \geq 0$. The number of Laguerre basis functions used $n$ is given in the second column. The relative difference to the option value computed by (NM-LS) with $M = 10$ million, $b^S = 3$ and $b^X = 1$ is indicated in the fourth column.
%
%\UnTableau{
%\begin{tabular}{|c|c|c|c|}
%	\hline
%		%& \multicolumn{2}{|c|}{NM-LS} & \multicolumn{2}{|c|}{M-LS} \\
%	$N_{\delta}$ & $n$ & (Lag-LS*) & Relative difference \\
%								& & & to (NM-LS) \\
%	\hline
%		5	& 4	& 3.530 & 0.13 \% \\
%		10&	7	& 4.278	& 0.28 \% \\
%			& 5	& MB & MB \% \\
%		15&	7	& 4.400	& 0.16 \% \\
%			& 5	& MB & MB \% \\
%		25&	7	& 4.084	& 0.10 \% \\
%			& 5	& MB & MB \% \\
%		35&	7	& 3.707	& 0.23 \% \\
%			& 5	& MB & MB \% \\
%	\hline
%	\end{tabular}
%}{Moving average options pricing with (Lag-LS*).}{tab-comp-NMLS-LSLag}{}{}

Figure \ref{valo-LS-Lag-LS-NM-call-delta50} presents the prices of
moving average call options computed by (Lag-LS*) and (NM-LS) when varying
$\delta$ from $0$ to $T$ with the same parameters as above. $7$
Laguerre basis functions were used with method (Lag-LS) as soon as
$N_{\delta} \geq 8$. For smaller $N_{\delta}$, we take $n = N_{\delta}
- 1$: $n$ must satisfy the condition $n \leq N_{\delta} - 1$ because otherwise the estimation of the conditional expectation at time $t_{N_{\delta}}$ leads to a degenerate linear system.
In the limit case when $\delta = T$, we retrieve the price of the
Asian option with payoff $\left(S_T - \frac{1}{T}\int_0^T S_t
  dt\right)^{+}$. For large averaging periods, the price that we obtain with $7$ Laguerre functions is about $0.30 \%$ above the benchmark value given by the (suboptimal) non Markovian approximation. \\
%

%\begin{figure}
%\centerline{\includegraphics[width=0.7\textwidth]{graphes/LS-NM_LS-Lag_linear_delta1-50.pdf}}
%\caption{Moving average option price as function of the averaging
%  window $\delta$.}
%\label{valo-LS-Lag-LS-NM-linear-delta50}
%\end{figure}

% SANS CAPTION et ZOOM pour CALL
\UneFigure{14}{13}{0}{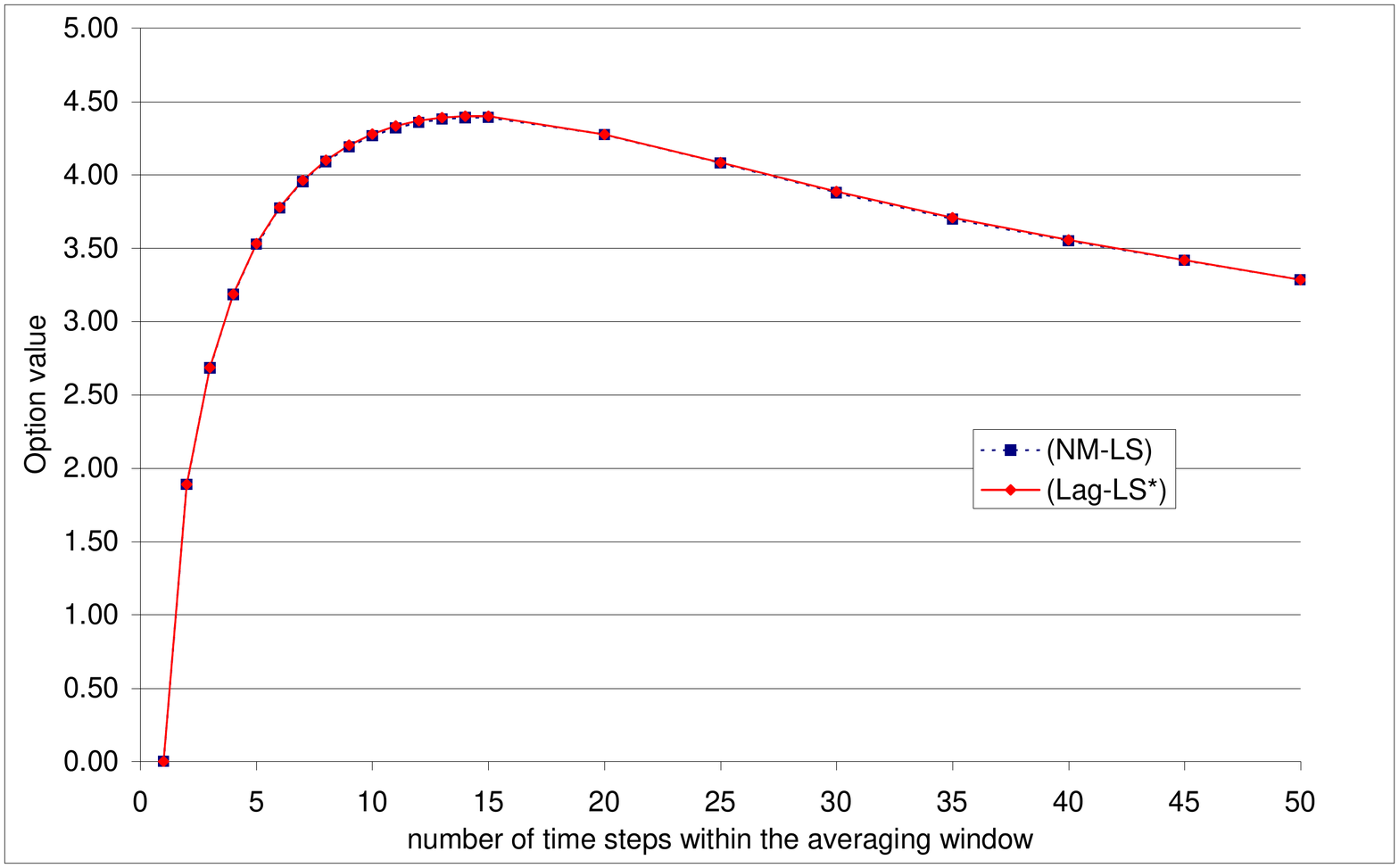}{Moving average option price as function of the averaging window $\delta$.}{valo-LS-Lag-LS-NM-call-delta50}{-1}{-0.8}

%The improved Laguerre-based method (Lag-LS*) seems thus to give
%accurate moving average option prices already with $n \geq 3$
%Laguerre basis functions (recall that the dimension of our Markovian
%approximation is equal to $(n+1)$), whatever the length of the
%averaging window. 

%%%%%%%%%%%%%%%%%%%%%%%%%%%%%%%%%%%%%%%%%%%%%%%%%%%%%%%%
%%%%%%%%%%%%%%%%%%%%%%%%%%%%%%%%%%%%%%%%%%%%%%%%%%%%%%%%
\paragraph{Moving average American options with time delay}

Consider now a moving average American option with time delay $l \geq 0$ whose value at time $0$ is:
\begin{eqnarray*}
\sup_{\tau \in \Tcurs_{[\delta + l, T]}} \E \left[ \phi \left(
    S_{\tau}, X_\tau \right) \right],\quad X_\tau = \frac{1}{\delta} \int_{\tau - l - \delta}^{\tau - l} S_{u} du.
%\label{def-pb-lag}
\end{eqnarray*}
With the same option characteristics and parameters as above and an averaging period equal to $\delta = 0.02$ (number of time steps $N_{\delta} = 5$), Figure \ref{valo-LS-Lag-call-delta5-lag45} presents the prices of delayed moving average call options computed by (Lag-LS*) and (NM-LS) when varying $l$ from $0$ to $T - \delta$. In the limit case when $l = T - \delta$ we retrieve the price of the Asian option with payoff $\left(S_T - \frac{1}{T - l}\int_0^{T-l} S_t dt\right)^{+}$. 

The relative difference between the option values given by (Lag-LS*) and (NM-LS) is significant (bigger than $5 \%$) for time lags such that $l \in [0.04, 0.152]$ (corresponding to $10 \leq N_{l} \leq 38$).
For example, when $l = 0.1$ (corresponding to $N_l = 25$ time steps), the relative difference is around $11 \%$. \\
%
%\begin{figure}
%\centerline{\includegraphics[width=0.7\textwidth]{graphes/LS-NM_LS-Lag_linear_lag0-45.pdf}}
%\caption{Price of the moving average option with time delay as
%  function of the lag $l$.}
%\label{valo-LS-Lag-linear-delta5-lag45}
%\end{figure}

% SANS CAPTION et ZOOM pour CALL
\UneFigure{13}{13}{0}{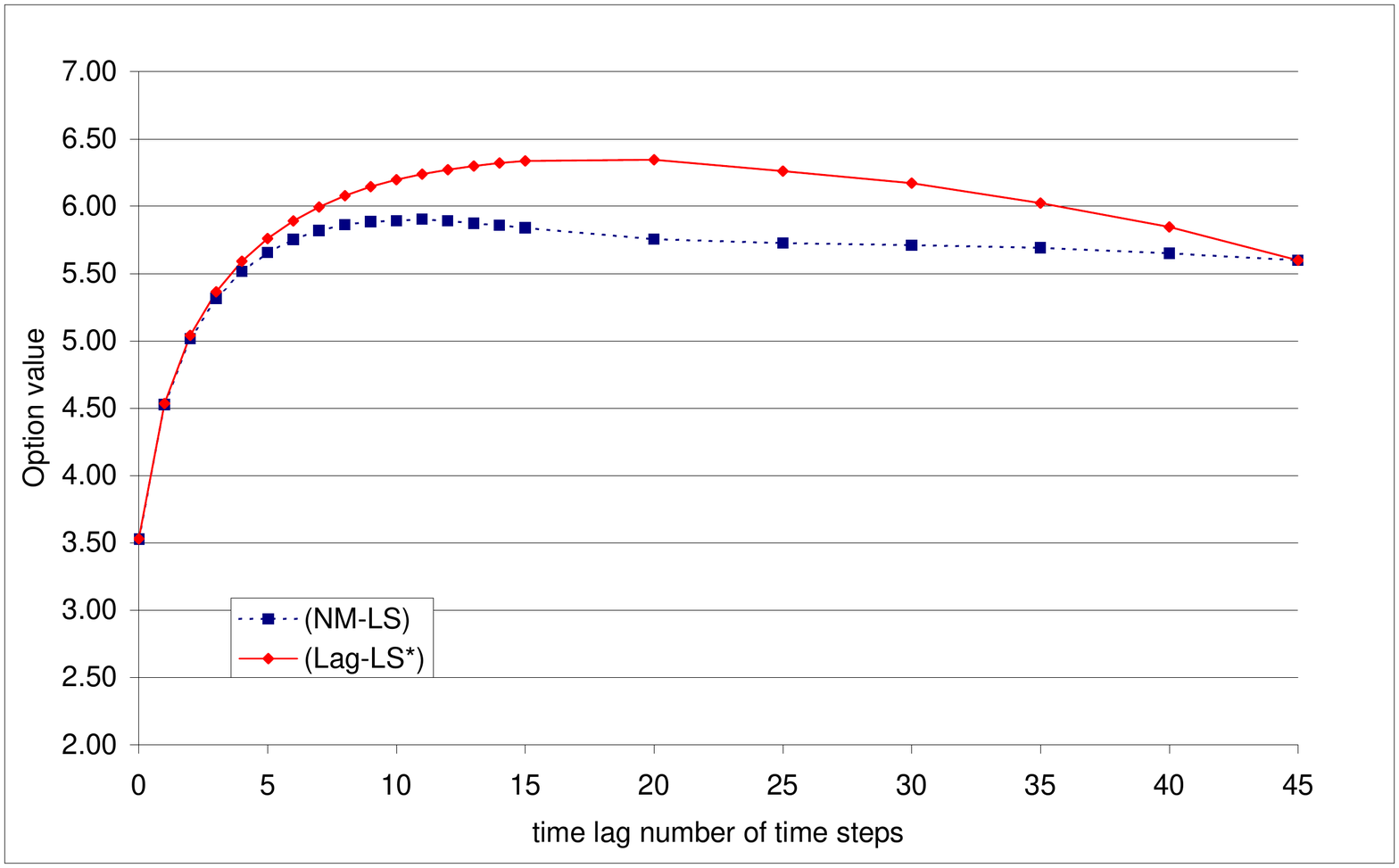}{Price of the moving average option with time delay as function of the lag $l$.}{valo-LS-Lag-call-delta5-lag45}{-1}{-0.8}
Now fix $l = 0.08$ ($N_l = 20$). As shown in Figure \ref{valo-LS-Lag-call-lag20}, when the averaging window increases this relative difference decreases. But it is still around $5 \%$ when $N_{\delta} = 15$. \\

%\begin{figure}
%\centerline{\includegraphics[width=0.7\textwidth]{graphes/LS-NM_LS-Lag_linear_lag20_fctdelta.pdf}}
%\caption{Price of the moving average option with time delay as
%  function of the averaging window $\delta$.}
%\label{valo-LS-Lag-linear-lag20}
%\end{figure}

% SANS CAPTION et ZOOM pour CALL
\UneFigure{13}{13}{0}{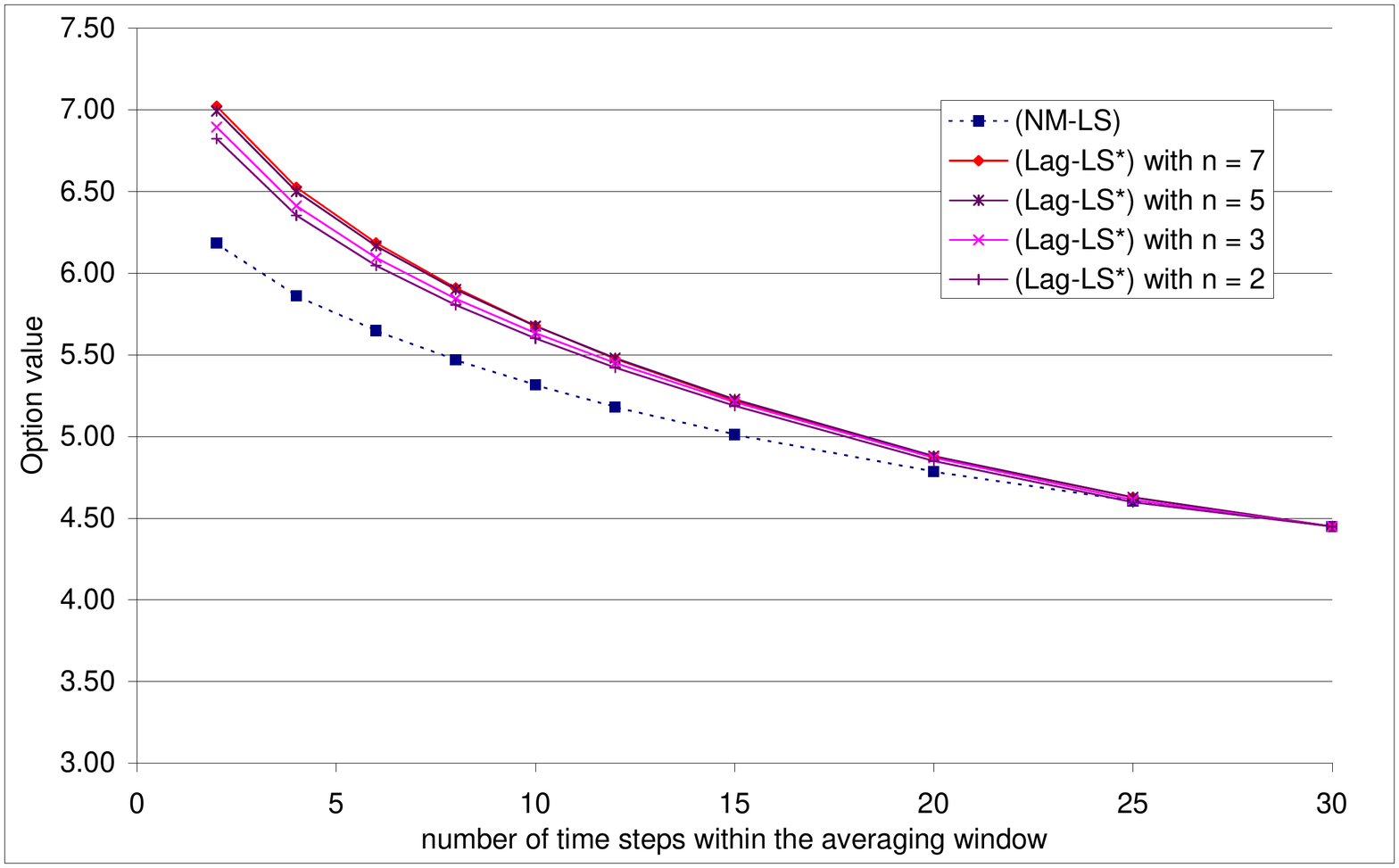}{Price of the moving average option with time delay\\as function of the averaging window $\delta$. }{valo-LS-Lag-call-lag20}{-1}{-0.8}

As expected, for pricing delayed moving average options, the non
Markovian approximate method (NM-LS) seems to be a worse approximation
than in the case without time lag: the error increases with the time
lag but decreases with the length of the the averaging period.
For large time lags and relatively small averaging period, our
improved Laguerre-based method (Lag-LS*) gives option prices up to $11
\%$ above the benchmark value given by the suboptimal approximation
(NM-LS) (cf. $N_l = 25$ on Figure
\ref{valo-LS-Lag-call-delta5-lag45}). For a good accuracy of
(Lag-LS*), the required number of Laguerre functions is however bigger
than in the case without time lag as explained at the end of Section \ref{sec-part-theo}.

%
%%%%%%%%%%%%%%%%%%%%%%%%%%%%%%%%%%%%%%%%%%%%%%%%%%%%%%%%
%%%%%%%%%%%%%%%%%%%%%%%%%%%%%%%%%%%%%%%%%%%%%%%%%%%%%%%%
%MB
\paragraph{Pricing of moving average Bermudean options: 
		a convergence rate improvement}
To compute the exact price of a moving average Bermudan option, one
can either use the classical method (M-LS) taking a sufficient number
of steps within the averaging window or use the Laguerre-based method
with sufficient number of state processes. In this last example we
illustrate the fact that our method (Lag-LS*) converges much faster than classical method (M-LS) with respect to the state dimension for pricing the same Bermudan option. 

Let us consider a moving average call option with maturity $T = 0.5$ and moving window $\delta = 0.1$.
Figure \ref{comp-CV-rate} provides a comparison between price values given by (Lag-LS*) for a time step $\Delta t = \frac{1}{80}$ when varying the number of Laguerre functions from $1$ to $7$ and by (M-LS) when varying the number of time steps within the averaging period from $2$ to $8$, that is $\Delta t = \frac{1}{20}, \frac{1}{30}, \ldots, \frac{1}{80}$ (the state dimension varies in both cases from $2$ to $8$). 
$M = 20$ million of Monte Carlo paths were used in both cases and $b^S=2$, $b^X=1$. \\

\UneFigure{15}{15}{0}{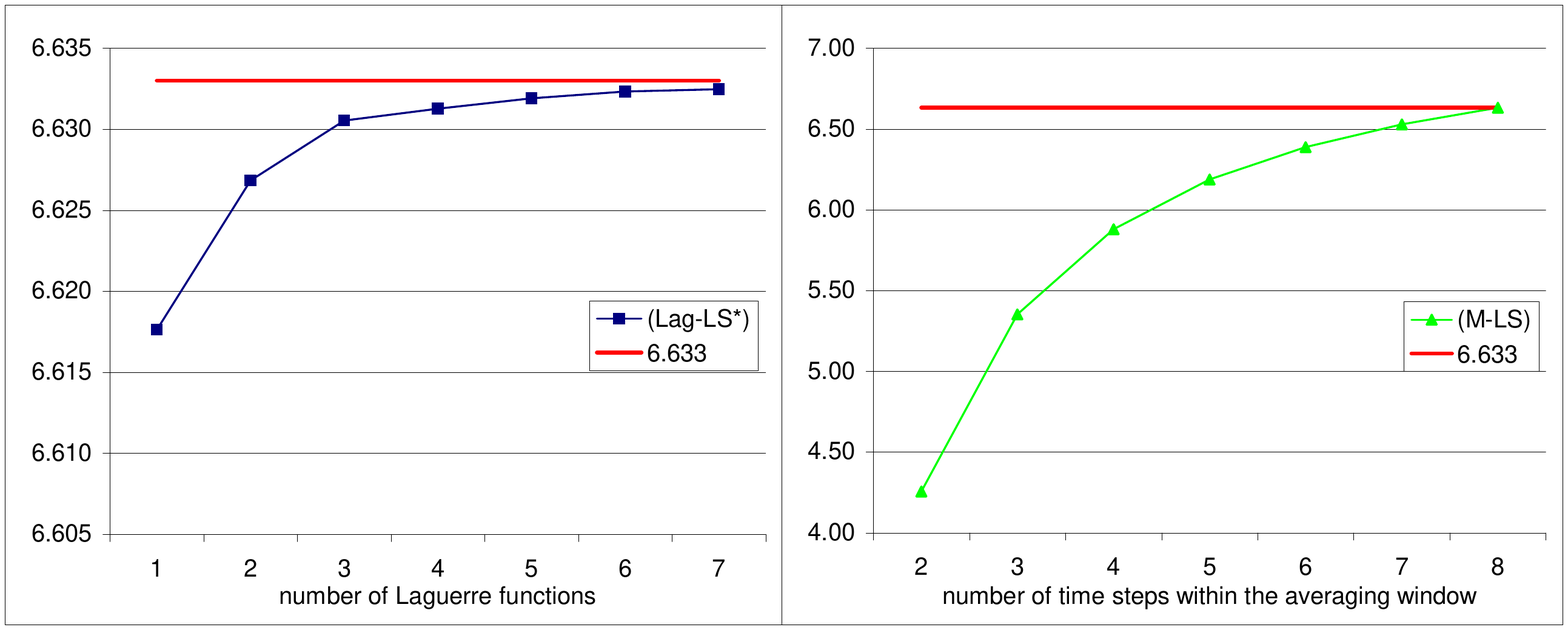}{Convergence of
  the improved Laguerre-based approximation and the benchmark method
					for pricing a Bermudan option.}{comp-CV-rate}{-1.2}{-3.5}

%%%%%%%%%%%%%%%%%%%%%%%%%%%%%%%%%%%%%%%%%%%%%%%%%%%%%%%%
% Biblio
%%%%%%%%%%%%%%%%%%%%%%%%%%%%%%%%%%%%%%%%%%%%%%%%%%%%%%%%

%%%%%%%%%%%%%%%%%%%%%%%%%%%%%%%%%%%%%%%%%%%%%%%%%%%%%%%%
% Annexe
%%%%%%%%%%%%%%%%%%%%%%%%%%%%%%%%%%%%%%%%%%%%%%%%%%%%%%%%
\appendix
%%%%%%%%%%%%%%%%%%%%%%%%%%%%%%%%%%%%%%%%%%%%%%%%%%%%%%%%
\section{Appendix}

%We provide in this appendix useful properties of Laguerre polynomials.

% %
%The following lemma sums up some useful properties of Laguerre polynomials. \eqref{prop-Lag-a} and \eqref{prop-Lag-c}  can be found for example in Lebedev \cite{Leb65} or Szegö \cite{Sze59} and \eqref{prop-Lag-b} is a direct consequence of \eqref{prop-rec-laguerre}.
% %
\begin{lem} The Laguerre polynomials $(P_{k})_{k \geq 0}$ belong to $\Ccurs^{\infty}\left( [0, +\infty ) \right)$ and:
 \begin{enumerate}[(i)]
 	\item \label{prop-Lag-a} $\forall k \geq 1, t P'_k(t) - k P_k(t) + k P_{k-1}(t) = 0$
 	\item \label{prop-Lag-b} $\forall k \geq 1, \frac{k}{t} \left( P_k(t) - P_{k-1}(t) \right) = - \sum^{k-1}_{i = 0} P_i (t)$
 	%\item \label{prop-Lag-c} $\forall k \geq 0, t P''_k(t) + (1 - t) P'_k(t) + k P_k(t) = 0$
 \end{enumerate}
 \label{lem-lag}
 \end{lem}
\begin{proof} \eqref{prop-Lag-a} can be found for example in Szegö \cite{Sze59} and \eqref{prop-Lag-b} is a consequence of \eqref{prop-rec-laguerre}.
\end{proof}

\begin{lem}\label{deriv.lm}
The definite integrals and derivatives of Laguerre functions can be
computed using the following formulas:
\begin{align}
\int_t^\infty e^{-s/2} P_n(s) ds &= 2e^{-t/2} P_n(t) + 4 e^{-t/2
}\sum_{k=1}^n (-1)^k P_{n-k} (t).  \\
\left(e^{-t/2} P_n(t)\right)' &= -\sum_{k=0}^{n-1} e^{-t/2} P_k(t) -
\frac{1}{2}e^{-t/2} P_n(t).
\end{align}
\end{lem}
\begin{proof}
This follows, after some computations, from the contour integral
representation of Laguerre polynomials:
$$
P_n(t) = \frac{1}{2\pi i} \oint \frac{e^{-\frac{ts}{1-s}}}{(1-s)s^{n+1}}ds.
$$
\end{proof}

\begin{lem}\label{integrals.lm}
The Laguerre functions and their integrals admit the following
representation in terms of Bessel functions:
\begin{align}
e^{-x/2}P_n(x) &= \sum_{k=0}^\infty A_k
\left(\frac{x}{\nu}\right)^{k/2} J_k(\sqrt{\nu x}),\quad \nu = 4n+2\\
I_n(x) := \int_0^x e^{-x'/2}P_n(x') dx' &=  2\sum_{k=0}^\infty A_k
\left(\frac{x}{\nu}\right)^{(k+1)/2} J_{k+1}(\sqrt{\nu x})\\
\int_0^x I_n(x') dx'&= 4\sum_{k=0}^\infty A_k
\left(\frac{x}{\nu}\right)^{(k+2)/2} J_{k+2}(\sqrt{\nu x}),
\end{align}
where $A_0=1$, $A_1=0$, $A_2 = \frac{1}{2}$ and other $A_i$-s satisfy the equation
$(m+2)A_{m+2} = (m+1)A_m - \frac{\nu}{2}A_{m-1}$. The series converge uniformly in $x$ on any compact interval.
\end{lem}
\begin{proof}
The first formula is from \cite{erdelyi.al.53}. The other two follow readily
using the integration formula for Bessel functions:
$$
\int_0^1 x^{\nu+1} J_\nu(ax)dx = a^{-1} J_{\nu+1}(a).
$$
\end{proof}

\begin{lem}\label{coefs.lm}
Let $\mu$ be a finite signed measure on $[0,\infty)$, with bounded
support which does not contain zero. Let $\{c_n\}$ denote the Laguerre
coefficients of the function $h(x):= \mu([x,\infty))$ and $\{A_n\}$
denote the Laguerre coefficients of the function $H(x):= \int_x^\infty
h(t)dt$. Then
$$
c_n = \Ocurs(n^{-3/4})\quad \text{and}\quad A_n = \Ocurs(n^{-5/4}).
$$
In addition, for $x>0$ fixed, $e^{-x/2}P_n(x) = \Ocurs(n^{-1/4})$.
\end{lem}
\begin{proof}
This result follows from Lemma \ref{integrals.lm}, using the
asymptotic expansion for Bessel functions
$$
J_n(x) = (\frac{1}{2}\pi x)^{-1/2} \cos\left(x-\frac{\pi}{2}n -
  \frac{\pi}{4}\right) + \Ocurs(x^{-3/2}),
$$
which holds uniformly \cite{erdelyi.al.53} on bounded domains outside
a neighborhood of zero.
\end{proof}

%%%%%%%%%%%%%%%%%%%%%%%%%%%%%%%%%%%%%%%%%%%%%%%%%%%%%%%%
%%%%%%%%%%%%%%%%%%%%%%%%%%%%%%%%%%%%%%%%%%%%%%%%%%%%%%%%
\end{document}